\documentclass[11pt, reqno]{amsart}
 \usepackage[foot]{amsaddr}
\usepackage{amssymb,amsmath,amsthm}
\usepackage{graphicx}
\usepackage{bbm} 
\usepackage{bm} 
\usepackage[shortlabels]{enumitem}
\usepackage{multirow}
\usepackage{adjustbox}
\usepackage[final]{showkeys}
\usepackage{todonotes}
\usepackage{mathtools}
\usepackage{url}
\usepackage{textcomp}
\usepackage{xfrac}
\usepackage{booktabs}
\usepackage{commath}

\usepackage{float}
\usepackage{tikz}
\usetikzlibrary{graphs, quotes}
\usepackage[font=small]{caption}
\usepackage{subcaption}
\usepackage{textcmds}
\usepackage{dsfont}
\usepackage{resizegather}

\usepackage[capitalise]{cleveref}

\newtheorem{thm}{Theorem}[section]
\newtheorem{cor}[thm]{Corollary}
\newtheorem{lemma}[thm]{Lemma}

\theoremstyle{definition}
\newtheorem{defn}[thm]{Definition}
\theoremstyle{remark}
\newtheorem{rem}[thm]{Remark}

\numberwithin{equation}{section}

\newcommand{\sgn}{\mathrm{sgn}}
\newcommand{\RR}{\mathbb{R}}

\newcommand{\QQ}{\mathbb{Q}}
\newcommand{\PP}{\mathbb{P}}

\newcommand{\cA}{\mathcal{A}}

\newcommand{\cF}{\mathcal{F}}

\newcommand{\pp}{\textup{\texttt{+}}}
\newcommand{\mm}{\textup{\texttt{-}}}

\newcommand{\wind}{\textup{wind}}  

\renewcommand{\abs}[1]{\left\vert#1\right\vert}
\renewcommand{\set}[1]{\left\{#1\right\}}

\newcommand{\Rplus}{\mathbb{R}_{\geqslant 0}}

\newcommand{\RN}[1]{%
  \textup{\uppercase\expandafter{\romannumeral#1}}%
}

\title[Term structure shapes in the Svensson family]{Term structure shapes and their consistent dynamics in the Svensson family}
\author{Martin Keller-Ressel, Felix Sachse}
\thanks{The authors acknowledge funding from the DFG-Project `Shapes of the Term Structure of Interest Rates'}
\address{Department of Mathematics, TU Dresden, Germany}
\email{martin.keller-ressel@tu-dresden.de, felix.sachse@tu-dresden.de}

\begin{document}
\begin{abstract}
We examine the shapes attainable by the forward- and yield-curve in the widely-used Svensson family, including the Nelson-Siegel and Bliss subfamilies. We provide a complete classification of all attainable shapes and partition the parameter space of each family according to these shapes. Building upon these results, we then examine the consistent dynamic evolution of the Svensson family under absence of arbitrage. Our analysis shows that consistent dynamics further restrict the set of attainable shapes, and we demonstrate that certain complex shapes can no longer appear after a deterministic time horizon. Moreover a single shape (either inverse of normal curves) must dominate in the long-run.
\end{abstract}
\maketitle

\section{Introduction}
The Nelson-Siegel and the Svensson family are parametric interpolation families for yield curves and forward curves, which are widely used by national banks and other financial institutions \cite{BIS05}. The Nelson-Siegel family \cite{NS87} expresses the forward curve as a linear combination of three basis functions, commonly associated to \emph{level}, \emph{slope} and \emph{curvature} in the form 
     \begin{align}\label{eq:nelson}
        \phi_{NS}(x) = \beta_0 + \beta_1\exp\left(-\frac{x}{\tau}\right) + \frac{\beta_2}{\tau}x\exp\left(-\frac{x}{\tau}\right).
    \end{align}
Svensson \cite{Sve94} argues that this family is not flexible enough to reproduce more complex shapes with multiple \emph{humps} and \emph{dips}\footnote{We use `hump' and `dip' synonymous with `local maximum' and `local minimum'.}, as they are frequently encountered in the market, and adds another curvature term with a different time-scale $\tau_2 \neq \tau_1$, resulting in 
    \begin{align}\label{eq:svensson}
        \phi_{S}(x) = \beta_0 + \beta_1\exp\left(-\frac{x}{\tau_1}\right) + \frac{\beta_2}{\tau_1}x\exp\left(-\frac{x}{\tau_1}\right) + \frac{\beta_3}{\tau_2}x\exp\left(-\frac{x}{\tau_2}\right).
    \end{align}
    A further interpolation family, due to Bliss \cite{Bli97}, is obtained by setting $\beta_2 = 0$ in the Svensson parametrization. Here, we are interested in the term structure shapes\footnote{In our terminology the term structure's \emph{shape} is characterized by the number and types of its local extrema; see also Table~\ref{tab:shape}.} that can be represented by the Svensson family of curves and by its subfamilies (Nelson-Siegel, Bliss). The shape of the term structure is a fundamental economic indicator and it encodes important information on market preferences for short-term vs. long-term investments, on expectations of central bank decisions and on the general economic outlook. Decreasing (`\texttt{inverse}') shapes of the term structure, for example, have been identified as signals of economic recessions, see \cite{bauer2018economic}. On the other hand, complex shapes with multiple local extrema have frequently been observed in both US and Euro area markets, see e.g. \cite{gurkaynak2007us} or \cite[App.~B]{diez2020yield}. Two exemplary yield curves from the Euro area are shown in \cref{fig:examples}. 
    
    Despite the wide adoption of the Svensson family, some of its basic properties relating to term-structure shapes have never been systematically studied, such as these simple questions:
    \begin{description}
    \item[Q1] Which shapes, exactly, can be represented by the Svensson family, and which cannot? 
    \item[Q2] What is the shape of the forward curve and the yield curve for a given parameter vector $\beta = (\beta_0 ,\beta_1, \beta_2, \beta_3)$?
    \item[Q3] Under a consistent dynamic evolution\footnote{Consistent in the sense of \cite{bjork1999interest, filipovic1999note,Fil00}.} of the Svensson family, which shapes are observed with strictly positive probability?
    \end{description}

 We will refer to the the first question as the \emph{classification problem}, the second question as the \emph{segmentation problem} (asking for a  segmentation of the parameter space according to shape) and the third question as the \emph{consistent dynamics problem}. While partial answers to these problems can certainly be given by numerical experiment, we are aiming for more general and systematic answers: We want to understand how the parameter space $\Theta = \set{\beta: (\beta_0,\beta_1,\beta_2,\beta_3) \in \RR^4}$ of the Svensson family decomposes into regions $R_1, \dotsc, R_k$, such that each region is associated to a certain shape of the forward curve or the yield curve, and how frequently these regions are visited by the corresponding stochastic factor process under a consistent dynamic simulation. Despite the apparent simplicity of questions \textbf{Q1}-\textbf{Q3}, it seems that such systematic results are currently missing from the literature\footnote{In models different from Svensson's, questions \textbf{Q1}-\textbf{Q3} have been answered for one-dimensional affine short rate models in \cite{keller-ressel2008yield, keller-ressel2018correction} and for the two-dimensional Vasicek model in \cite{diez2020yield, KR21, KRS23}.}, and we intend to close this gap with this paper.
 
 \begin{table}[hbtp]
\begin{center}
\begin{tabular}{p{3cm}p{5cm}p{3cm}} 
\toprule
Shape of the term structure & Description & Sign sequence of derivative\\ 
\midrule 
\texttt{normal (n)} & strictly increasing & $\pp$\\
\texttt{inverse (i)} & strictly decreasing & $\mm$\\
\texttt{humped (h)} & single local maximum & $\pp \mm$\\
\texttt{dipped (d)} & single local minimum & $\mm \pp$\\
\texttt{hd} & hump-dip, i.e. local maximum followed by local minimum & $\pp \mm \pp$\\
\texttt{dh}, \texttt{hdh}, etc. & further sequences of multiple `dips'  and `humps'  & $\dotsc$ \\
\bottomrule
\end{tabular}
\end{center}
\caption{Shapes of the term structure; reproduced from \cite{KRS23}.\label{tab:shape}}
\end{table}

\begin{figure}
        
            \centering
            \begin{subfigure}{0.8\textwidth}
            \includegraphics[width=\textwidth]{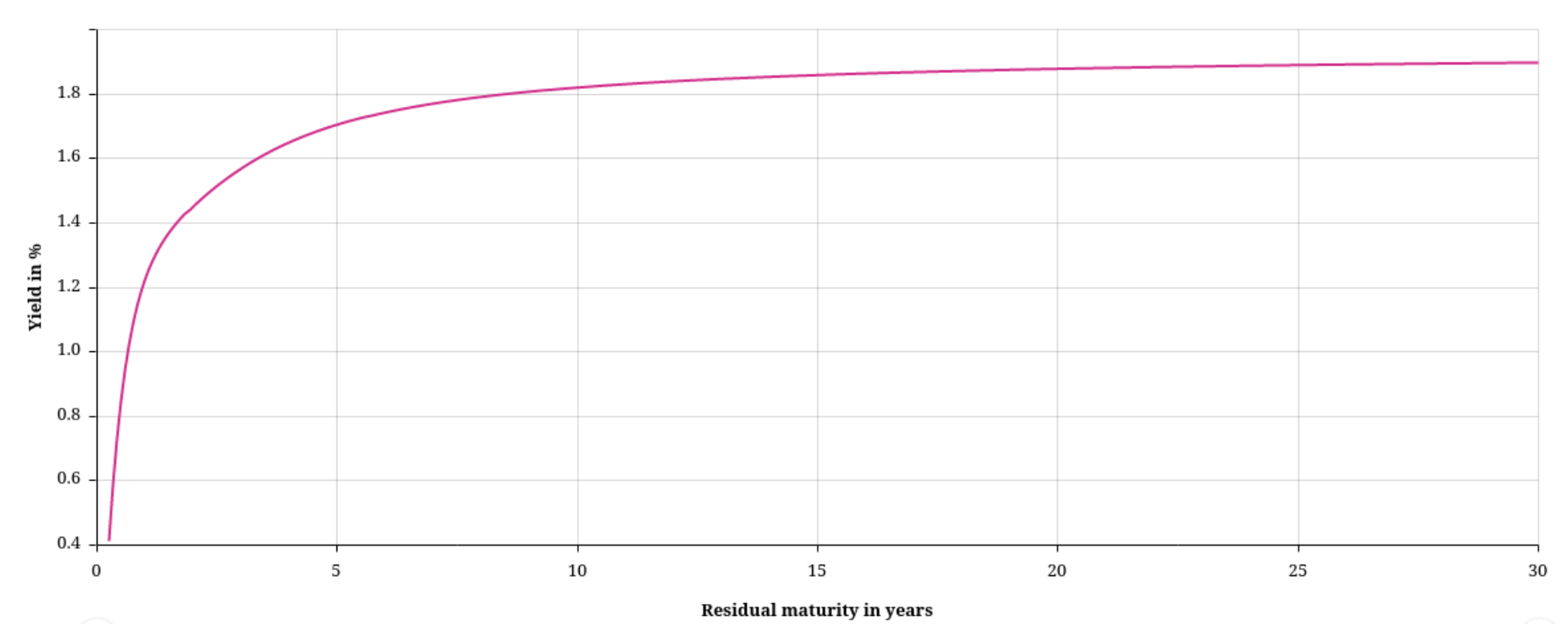}
            \caption{A \texttt{normal} yield curve}
            \label{fig:yield_curve_normal}
            \end{subfigure}

                 \begin{subfigure}{0.8\textwidth}
 
            \includegraphics[width=\textwidth]{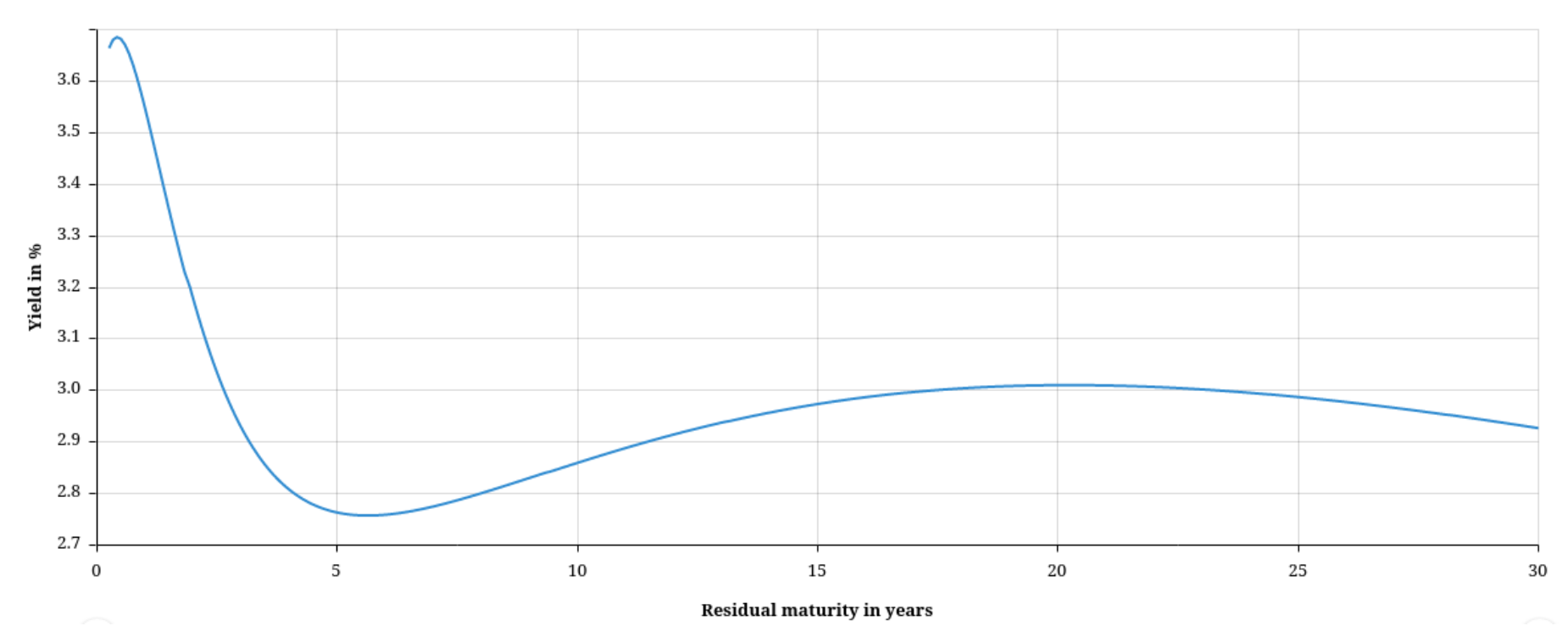}
            \caption{A yield curve of shape \texttt{hdh}}
            \label{fig:yield_curve_hdh}
            \end{subfigure}
            \caption{Two differently shaped yield curves for AAA-rated bonds in the Euro area. Plots generated through the ECB's website \protect\url{https://www.ecb.europa.eu} for (A) Sep 16, 2022 and (B) Sep 27, 2023.} \label{fig:examples}
\end{figure}

 \section{Overview of results}
 Before providing mathematical details and methods, we give an overview of our results. First, as a simple observation, the Svensson family (with $\beta_3 \neq 0$) can be reparametrized as
    \begin{align}\label{eq:svensson_2}
        \phi_{S}(x) = \beta_0 + \beta_3 \left( \gamma_\RN{2} \exp\left(-\frac{x}{\tau_1}\right) + \frac{\gamma_\RN{1}}{\tau_1}x\exp\left(-\frac{x}{\tau_1}\right) + \frac{1}{\tau_2}x\exp\left(-\frac{x}{\tau_2}\right) \right),
    \end{align}
 where we have set
 \begin{equation}\label{eq:gamma_coord}
 \gamma = (\gamma_\RN{1}, \gamma_\RN{2}) = \left(\frac{\beta_2}{\beta_3}, \frac{\beta_1}{\beta_3}\right) \in \RR^2.
 \end{equation}
As an additive shift, $\beta_0$ has no influence on the term structure's shape. Moreover, the scaling factor $\beta_3$ will only influence the shape through its sign. Thus, to answer questions \textbf{Q1}-\textbf{Q3}, the effective parameter space can be reduced to the two-dimensional space
 \[\Theta' = \set{\gamma: \gamma \in \RR^2}.\]
Furthermore, it will turn out that two \emph{phase transitions} in the Svensson family's behaviour will take place, depending on the ratio
\begin{equation}\label{eq:ratio}
 r = \frac{\tau_1}{\tau_2}
 \end{equation}
 of time-scales. 
 \begin{defn}\label{def:regime}
 We distinguish the following regimes:
             \begin{itemize}
                                 \item $r > 1$: {\bf scale-regular} (sr),
                    \item $r \in [1/3,1)$: {\bf weakly scale-inverted} (wsi), and
                    \item $r \in (0,1/3)$: {\bf strongly scale-inverted} (ssi).
                \end{itemize}
                \end{defn}
                
\subsection{The classification and segmentation problems}
Our first main results answer the question \textbf{Q1} of \emph{shape classification}. We start with a basic result, which is easy to obtain from arguments based on Tchebycheff systems, see Sec.~\ref{sub:TP}: 
\begin{thm}\label{thm:basic}The forward curve and the yield curve have
\begin{enumerate}
\item at most a single local extremum in the Nelson-Siegel family,
\item at most two local extrema in the Bliss family, and
\item at most three local extrema in the Svensson family.
\end{enumerate}
\end{thm}
Note that this result gives only necessary, but no sufficient conditions for a particular shape to appear. To give a more refined picture, we call a term structure shape (in the sense of Table~\ref{tab:shape}) \emph{attainable}, if there exists a parameter $\beta \in \Theta$ such that the corresponding yield curve or forward curve in the Svensson family exhibits this shape. The same terminology also applies in the Bliss family, the Nelson-Siegel family, or under other parameter restrictions, e.g. when the sign of $\beta_3$ is restricted.

                                \begin{thm}\label{thm:main}The following statements hold for both the yield curve and the forward curve in the Svensson family:
                \begin{enumerate}[(a)]
                              \item In the scale-regular regime the shapes \texttt{normal, inverse, dipped and humped} are attainable for any $\beta_3 \in \RR$. If $\beta_3 > 0$ also the shapes \texttt{hd} and \texttt{hdh} are attainable; if $\beta_3 < 0$ also the shapes \texttt{dh} and \texttt{dhd}.
                              \item In the weakly scale-inverted regime the shapes \texttt{inverse}, \texttt{humped} and \texttt{dh} are attainable if $\beta_3 > 0$. The shapes \texttt{normal}, \texttt{dipped} and \texttt{hd} are attainable if $\beta_3 < 0$. 
\item In the strongly scale-inverted regime the shapes \texttt{inverse}, \texttt{humped}, \texttt{dh} and \texttt{hdh} are attainable if $\beta_3 > 0$. The shapes \texttt{normal}, \texttt{dipped}, \texttt{hd}, \texttt{dhd} are attainable if $\beta_3 < 0$. 
\end{enumerate}
 \end{thm}
In this attainability result, it is not necessary to distinguish between yield curve and forward curve. The \emph{segmentation of the parameter space}, however, will differ between these cases. We illustrate regime (a) in \cref{fig:scale_regular}, regime (b) in \cref{fig:wsi}, and regime (c) in \cref{fig:ssi}. All analytic results on the segmentation problem can be found in \cref{sec:proofs}. Note that the case $\beta_3 = 0$ excluded in \cref{thm:main} corresponds to the Nelson-Siegel family and is treated in \cref{sub:NS}. Further results on the Bliss family ($\beta_2 = 0$) can be found in \cref{sub:Bliss}. Finally, we remark, that even when a particular shape is attainable, it may be attainable only for a very narrow range of parameters. See for example the tiny spot in \cref{fig:b3g0_y} that corresponds to three local extrema (\texttt{hdh}).

\begin{rem}\label{rem:polarity}
 Note that the segmentation of $\Theta'$ is independent of $\beta_3$, but changing the sign of $\beta_3$ will reverse the `polarity' of each region, in the sense that local minima and local maxima of the forward curve are interchanged. 
 \end{rem}
 \begin{figure}
        
            \centering
            \begin{subfigure}{0.85\textwidth}
            \includegraphics[width=\textwidth]{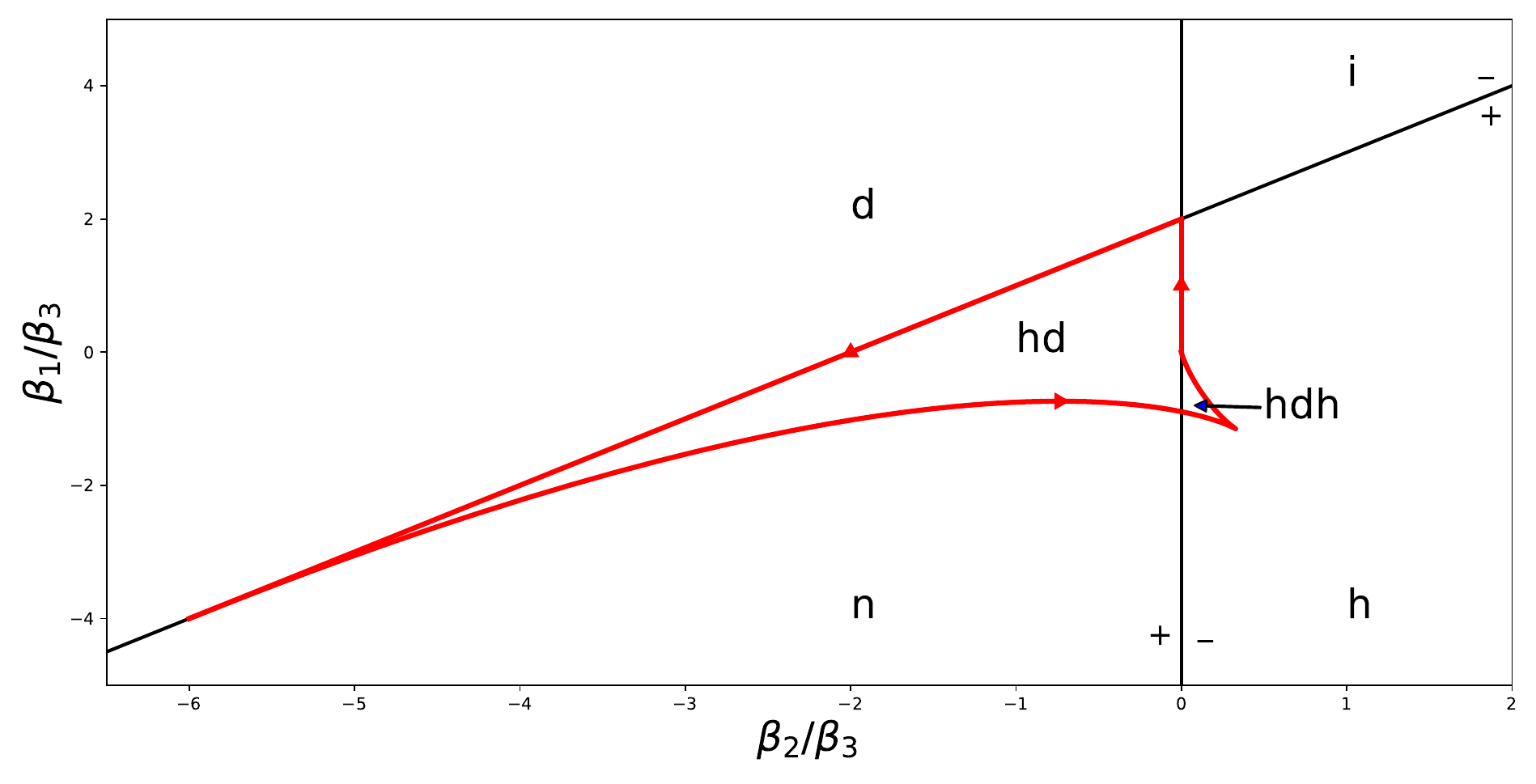}
            \caption{Shapes of the forward curve}
            \label{fig:b3g0}
            \end{subfigure}

                 \begin{subfigure}{0.85\textwidth}
 
            \includegraphics[width=\textwidth]{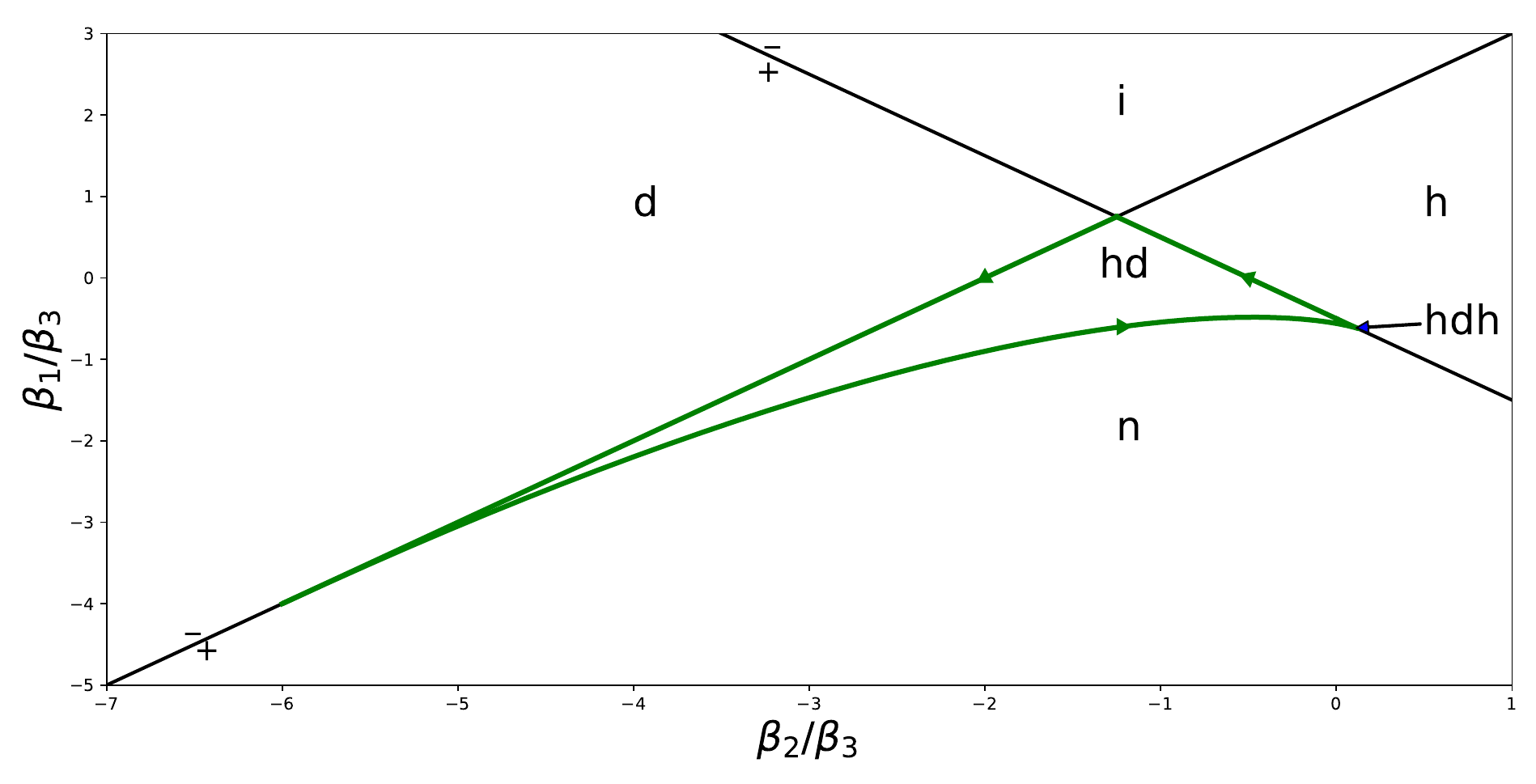}
            \caption{Shapes of the yield curve}
            \label{fig:b3g0_y}
            \end{subfigure}
            \caption{Shapes of the forward curve (A) and the yield curve (B) for different regions of the parameter space $\Theta'$ of the Svensson family in the scale-regular regime $r > 1$. Red and green curves indicate the augmented envelope $\hat \eta$ (see Def.~\ref{defn:augmented}). Parameters used are $\tau_1 = 1, \tau_2 = 1/2$ and region labels (see Table~\ref{tab:shape}) correspond to the case $\beta_3 > 0$ . For $\beta_3 < 0$ the plot stays the same, but region labels must be changed, see Remark~\ref{rem:polarity}} \label{fig:scale_regular}
\end{figure}

\subsection{Empirical observations}
The ECB publishes yield curves and forward curves, modelled with the Svensson family, on a daily basis. To complement our theoretical results, we present a brief overview of term structure regimes and shapes observed for AAA-rated bonds in the Euro Area from September 06th, 2004 to September 24th, 2024.
	
	In \cref{fig:regimes} we show how the value of $r$, see \eqref{eq:ratio}, changes over time. The plot shows that every regime from \cref{def:regime} is attained. More precisely, the regimes are attained with the following relative frequencies.
	\begin{center}
		\begin{tabular}{rrrrrrrr}
			$\beta_3 > 0$:$\quad$ 	&sr: & 2.2\%, & wsi: & 2.0\%, & ssi: & 6.7\%,\\
			$\beta_3 < 0$:$\quad$	&sr: & 46.0\%, & wsi: & 35.0\%, & ssi: & 8.1\%.
		\end{tabular}
	\end{center}
	\begin{figure}
			\includegraphics[width=\textwidth]{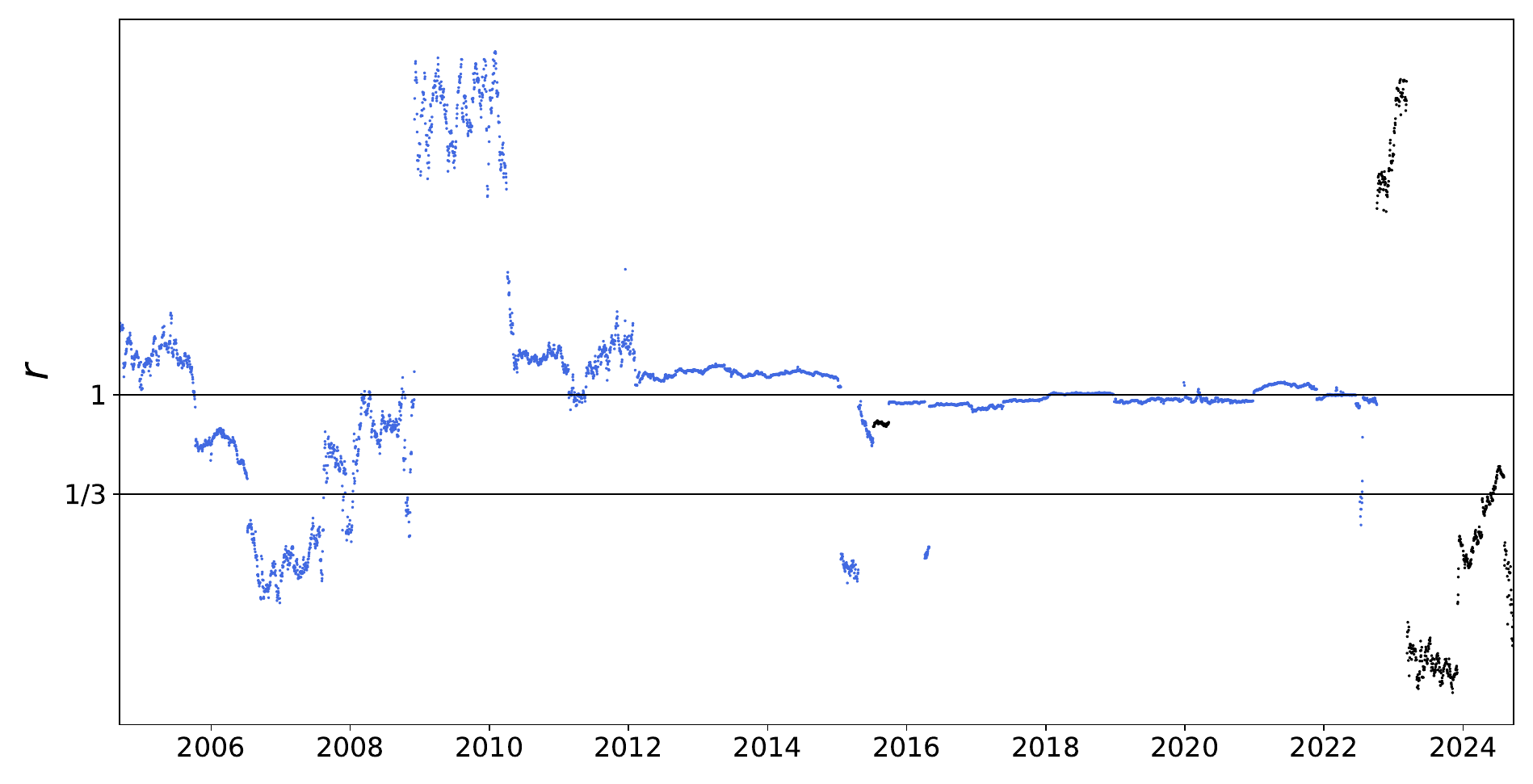}
			\caption{Values of $r = \tau_1/\tau_2$ visiting different regimes (see Def.~\ref{defn:augmented}) in the Svensson family for Euro Area AAA-rated bond data over time. Dots are colored black if $\beta_3 > 0$ and blue if $\beta_3 < 0$.}
			\label{fig:regimes}
	\end{figure}
	
	In \cref{fig:shapes} we show how the shape of the yield curve and the forward curve change over time. Every shape with at most 3 local extrema, except the \texttt{inverse} shape, is attained. More precisely, the yield curve attains the shapes with the following relative frequencies
	\begin{center}
	\begin{tabular}{rrrrrrrr}
		\texttt{n}:  & 19.2\%, &\texttt{h}:  & 7.9\%, &\texttt{hd}:  & 13.0\%, &\texttt{hdh}:  & 5.0\%\\
		\texttt{i}:  & 0.0\%, &\texttt{d}:  & 25.0\%, &\texttt{dh}:  & 29.4\%, &\texttt{dhd}:  & 0.5\%,
	\end{tabular}
	\end{center}
	and for the forward curve the values change to
	\begin{center}
	\begin{tabular}{rrrrrrrr}
	    	\texttt{n}:  & 11.6\%, &\texttt{h}:  & 10.1\%, &\texttt{hd}:  & 18.1\%, &\texttt{hdh}:  & 5.3\%\\
	    	\texttt{i}:  & 0.0\%, &\texttt{d}:  & 13.6\%, &\texttt{dh}:  & 39.7\%, &\texttt{dhd}:  & 1.6\%.
	\end{tabular}
	\end{center}
	\begin{figure}
		\includegraphics[width=\textwidth]{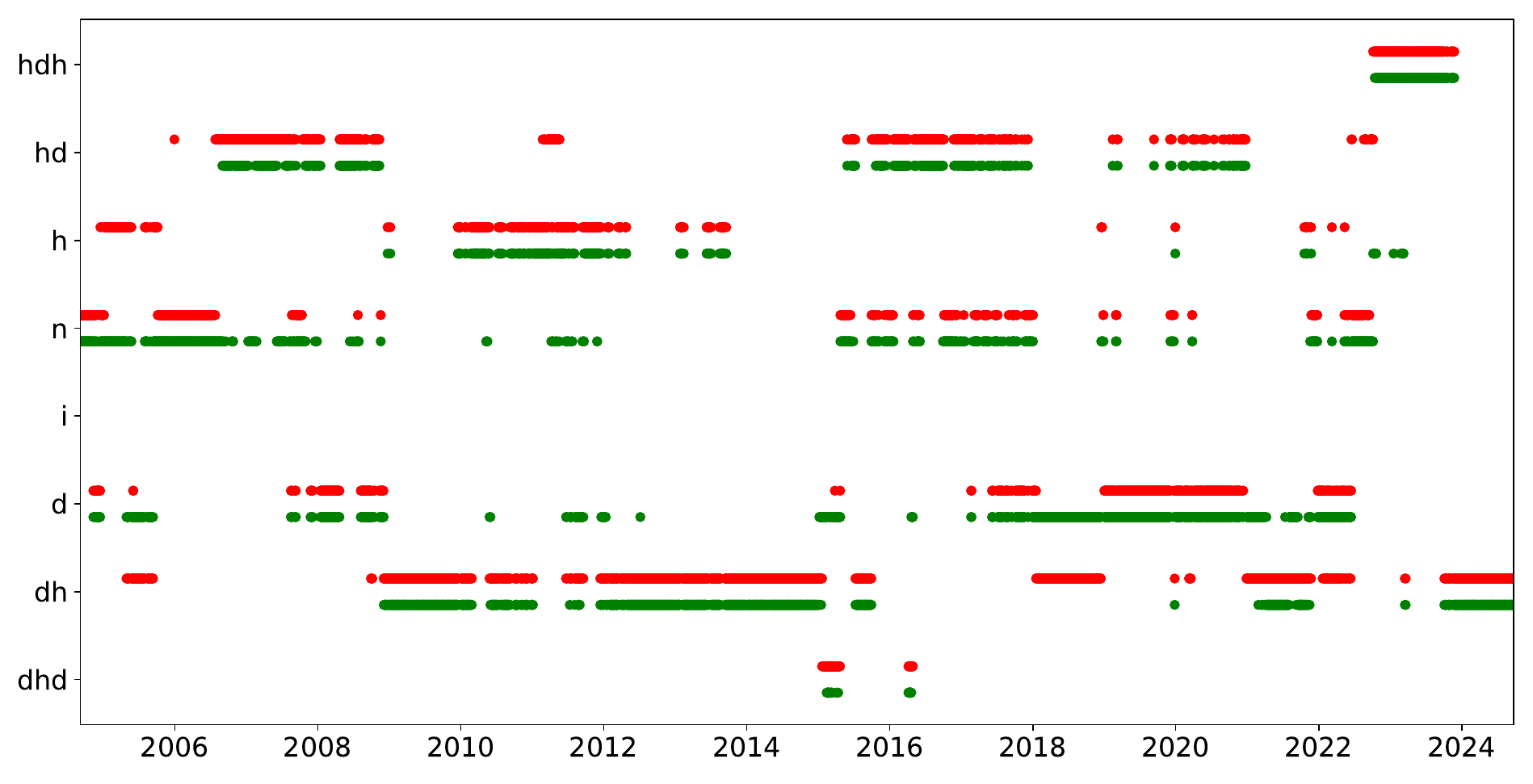}
		\caption{Shapes of the forward curve (red) and yield curve (green) for Euro Area AAA-rated bond data over time. Abbreviations are explained in Table~\ref{tab:shape}}
		\label{fig:shapes}
	\end{figure}

 \subsection{Consistent dynamic evolution of term-structure shapes}\label{sub:consistency_preview}
 The question of the \emph{consistency} of a family of parametrized forward rate curves with the arbitrage-free dynamic evolution of interest rates has been initiated by Bj\"ork and Christensen in \cite{bjork1999interest} and has been answered in the particular cases of the Nelson-Siegel and the Svensson family by Filipovic in \cite{filipovic1999note} and \cite{Fil00}. For the consistency problem, we consider a dynamic version of the Svensson family \eqref{eq:svensson}, where its extended parameter vector $\tilde \beta = (\beta_0, \beta_1, \beta_2, \beta_3, \tau_1, \tau_2)$ is `set in motion' by replacing it by an Ito process $(\tilde \beta(t))_{t \ge 0}$ of the form
 \begin{equation}\label{eq:ito}
 d\tilde \beta(t) = b_t\,dt + \sigma_t dW_t,
 \end{equation}
 where $W$ is a $d$-dimensional ($d \ge 1$) Brownian motion on a standard filtered probability space $(\Omega, \cA, (\cA_t)_{t \geq 0}, \QQ)$, and $b, \sigma$ are progressively measurable processes of suitable dimensions (see~\cite{Fil00} for details). The process $\tilde \beta(t)_{t \geq 0}$ is called \emph{consistent} with the Svensson family, if zero-coupon bond prices, discounted with the money-market account, are $\QQ$-martingales, ensuring absence of arbitrage in the resulting interest rate model. It turns out, that the scope of Ito-processes \eqref{eq:ito}, which fulfill the requirement of consistency, is substantially restricted, and can be essentially reduced to a single-factor short rate model with time-dependent coefficients (again, we refer to \cite{Fil00} and \cite[Ch.~9]{Fil09} for details). Here, in answering question \textbf{Q3}, we are interested in the restrictions on the shapes of term-structure curves imposed by the consistency requirement, and obtain the following result:
 \begin{thm}\label{thm:main_consistency}
 In the consistent dynamic Svensson-model with initial forward curve given by \eqref{eq:svensson} the following holds:
 \begin{enumerate}
\item  If $\beta_2 > 0$ then there is a time-horizon
\begin{equation}\label{eq:T_dagger_fw}
 T^{\rm f}_\dagger = \max\left(-\tfrac{5}{2} + \log(\tfrac{4 \beta_3}{\beta_2}), 0\right),
 \end{equation}
such that the forward curve attains each of the shapes \texttt{inverse}, \texttt{humped} and \texttt{hdh} with strictly positive probability at any time $t < T_\dagger$, and each of the shapes \texttt{inverse}, \texttt{humped} with strictly positive probability at any time $t > T_\dagger$. All other shapes are almost surely never attained.
\item  If $\beta_2 < 0$ then there is a time-horizon
\begin{equation}\label{eq:T_star_fw}
 T^{\rm f}_\star = \max\left(\log\left(\tfrac{6 \beta_3}{|\beta_2|}\right), 0\right),
 \end{equation}
such that the forward curve attains each of the shapes \texttt{normal}, \texttt{dipped} and \texttt{hd} with strictly positive probability at any time $t < T_\star$, and each of the shapes \texttt{normal}, \texttt{dipped} with strictly positive probability at any time $t > T_\star$. All other shapes are almost surely never attained.
 \end{enumerate}
 \end{thm}
This result shows that more complex initial shapes, like \texttt{hdh} or \texttt{hd} are lost after a deterministic time-horizon under any consistent dynamic evolution of the Svensson family. An analogous result for the yield curve is given by \cref{thm:consistency_yield}, the results are illustrated in \cref{fig:shape-flow_chart_y}. As a direct corollary of \cref{thm:main_consistency}, we obtain:
\begin{cor}\label{cor:trapped}
The forward curve in the consistent dynamic Svensson-model can evolve only within one of the disjoint sets of shapes
 \[\{\texttt{inverse}, \texttt{humped}, \texttt{hdh}\}, \qquad \{\texttt{normal}, \texttt{dipped},  \texttt{hd}\},\]
 but cannot transition from one set to the other, almost surely with respect to $\mathbb{Q}$.
\end{cor}
Since the above results are statements about sets of probability zero (and their complements) it does not matter whether the risk-neutral measure $\QQ$ or any other equivalent measure $\PP$ is considered. In addition to these results, we show in Section~\ref{sub:frequency} how to calculate the exact risk-neutral probability of each attainable shape under consistent dynamic evolution of the Svensson family. Of particular interest is the following ergodicity result, which shows that almost surely, a \emph{single shape} will dominate the dynamic evolution in the long-run.
 
         	\begin{thm}\label{thm:ergodic}  Let $\mathsf{S}_t(\omega)$ denote the shape of either the yield- or the forward-curve at time $t > 0$ in the consistent dynamic Svensson-model.
	\begin{enumerate}
	\item If $\beta_2 > 0$, then
	\begin{equation}\label{eq:ergodic1}
	\lim_{t \to \infty }\mathbf{1}\set{\mathsf{S}_t(\omega) \text{ is } \texttt{inverse}} = 1, \qquad \QQ-\text{a.s.}
	\end{equation}
		\item If $\beta_2 < 0$, then
	\begin{equation}\label{eq:ergodic2}
		\lim_{t \to \infty }\mathbf{1}\set{\mathsf{S}_t(\omega) \text{ is } \texttt{normal}} = 1, \qquad \QQ-\text{a.s.}
	\end{equation}

	\end{enumerate}
         	\end{thm}
	
	The above results are illustrated in \cref{fig:shape-flow_chart_fw} for the forward curve and in \cref{fig:shape-flow_chart_y} for the yield curve.
	
	    \begin{figure}
		\begin{tikzpicture}
  \matrix[row sep=1mm,column sep=0.3cm] {

     \node[right] (h1) {\footnotesize $\phantom{t=0}$}; & & \node[above] (h2) {\footnotesize $t = T_\dagger^{\rm f}$}; & & \node (h3) {}; & \\
      & \node (n1)  {$\{\texttt{i},\texttt{h},\texttt{hdh}\}$}; &  & \node (n2) {$\{\texttt{i},\texttt{h}\}$}; &  & \node (n3)  {$\{\texttt{i}\}$};\\
 & & \node (h2up) {}; & & & \\
    \node[left] (n0)  {Calibration}; &  & & & &  \\
 & & \node (h5down) {}; & & & \\    
        & \node (n4)  {$\{\texttt{n},\texttt{d},\texttt{hd}\}$};  &  & \node (n5) {$\{\texttt{n},\texttt{d}\}$};  & & \node (n6)  {$\{\texttt{n}\}$};\\
           \node[below right] (h4) {\footnotesize $t=0$}; & & \node[below] (h5) {\footnotesize $t = T_*^{\rm f}$}; & & \node[below] (h6) {\footnotesize $t \to \infty$}; & \\
  };
  \graph {
  (n0) ->["{\footnotesize $\beta_2 > 0$}"] (n1) -> (n2) -> (n3);
  (n0) ->["{\footnotesize $\beta_2 < 0$}" below left] (n4) -> (n5) -> (n6);
  (h1) --[dashed] (h4); 
  (h2) --[dashed] (h2up); 
  (h5) --[dashed] (h5down); 
  (h3) --[dashed] (h6); 
  };
\end{tikzpicture}
    	\caption{This diagram shows the shapes that are attained with strictly positive probability for the forward curve in the consistent Svensson model as time $t$ progresses, illustrating Thm.~\ref{thm:main_consistency}, Cor.~\ref{cor:trapped} and Thm.~\ref{thm:ergodic}. The branch $\beta_2 > 0$ or $\beta_2 < 0$ is selected when the model is calibrated to the initial term structure. Note that each of the times $T_\dagger^{\rm f}, T_*^{\rm f}$ may be zero, in which case the corresponding part of the diagram collapses.}
    	\label{fig:shape-flow_chart_fw}
    \end{figure}
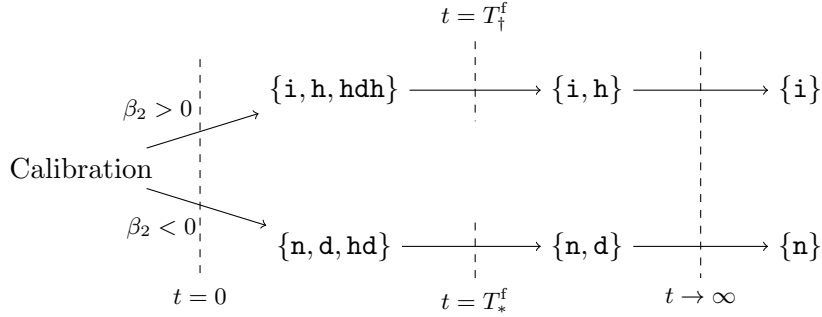

The remainder of the paper is structured as follows: In Section~\ref{sec:preliminaries} we discuss some preliminary results on the relation between yield- and forward-curves, on the Nelson-Siegel subfamily and on total positivity in the Svensson family. In Section~\ref{sec:envelope}, we review our main mathematical tool, the \emph{envelope method}, which was introduced in \cite{KRS23} in the context of the two-factor Vasicek model. In Section~\ref{sec:proofs}, we provide the proofs of our main results on the classification and the segmentation problem (\textbf{Q1},\textbf{Q2}), as well as some refinements of the theorems presented above. In Section~\ref{sec:consistency}, we do the same for the consistency problem (\textbf{Q3}). Some auxiliary results are relegated to Appendix~\ref{app}.

\section{Preliminaries and basic results}\label{sec:preliminaries}
\subsection{Yield curve vs. forward curve}Let $f(x,\beta)$ denote the forward curve in dependency on time-to-maturity $x$ and the parameter vector $\beta$. Then the associated yield curve is given by
\begin{equation}\label{eq:forward_to_yield}
y(x,\beta) = \frac{1}{x}\int_0^x f(y,\beta)dy, 
\end{equation}
and it can be considered as a running average (in time-to-maturity) of the forward curve. Consequently, the shape of the yield curve will be `smoother' than the shape of the forward curve. In fact, we have the following result from \cite{KR21}:\footnote{The result was shown in \cite{KR21} in the context of the two-factor Vasicek model, but only uses the relation \eqref{eq:forward_to_yield} between forward- and yield-curve}

\begin{thm}[{See \cite[Thm.~3.4]{KR21}}]\label{thm:fw_to_yield}
The following holds for the relation between yield curve and forward curve:
\begin{enumerate}
\item the initial slope of yield- and forward curve has the same sign;
\item the number of local extrema of the yield curve is less or equal to the number of local extrema of the forward curve;
\item if the number of local extrema is the same, the also the type (dip/hump) coincides.
\end{enumerate}
\end{thm}
For this reason, many results obtained for the forward curve have immediate consequences for the yield curve. Other results, for example on the segmentation of the parameter space, have to be shown separately for yield curve and forward curve. However, usually the same method can be used for both with minor adaptations. Using the Svensson-parametrization \eqref{eq:svensson} for the forward curve, one immediately sees that the corresponding yield curve 
    \begin{align}\label{eq:svensson_yield}
        \Phi_S(x) = &\beta_0 + \frac{\beta_1 \tau_1}{x}\left(1 - e^{-x/\tau_1}\right) + \beta_2\left(\frac{\tau_1}{x}\left(1 - e^{-x/\tau_1}\right) - e^{-x/\tau_1}\right) + \\ &+ \beta_3 \left(\frac{\tau_2}{x}\left(1 - e^{-x/\tau_2}\right) - e^{-x/\tau_2}\right) \notag
    \end{align}
will still be a linear function of the parameter vector $\beta$ with only the basis functions transformed by \eqref{eq:forward_to_yield}. Generally, most of our results will be obtained first for the forward curve and then be transferred to the yield curve; some useful Lemmas are given in Appendix~\ref{app}.

\subsection{The Nelson-Siegel family}\label{sub:NS}
For the Nelson-Siegel family, questions \textbf{Q1} and \textbf{Q2} can be solved directly, using elementary means. Looking at the forward curve, it is sufficient to take the time-derivative of \eqref{eq:nelson} to obtain
        \begin{align*}
            \phi_{NS}'(x) = \frac{1}{\tau}\left(\beta_2\left(1 - \frac{x}{\tau}\right) - \beta_1\right)\exp\left(-\frac{x}{\tau}\right).
        \end{align*}
Clearly, this derivative is zero at $x \in (0,\infty)$ if and only if
        \begin{align*}
            x = \tau\left(1 - \frac{\beta_1}{\beta_2}\right)
        \end{align*}
and, if it has a zero, it also changes sign at this point. Together with \cref{thm:fw_to_yield} this leads to the following simple result: 
        \begin{thm}\label{thm:Shapes_basic}
\label{item:Nelson-Siegel-forward}In the Nelson-Siegel family only the shapes \texttt{normal}, \texttt{inverse}, \texttt{humped} and \texttt{dipped} are attainable for the forward curve. The same holds true for the yield curve. The parameter space segmentation is given in Table~\ref{tab:Nelson-Siegel}.
    \end{thm}
    
         \begin{table}
            \begin{tabular}{ccc}
                \toprule
                shape & region (forward curve) & region (yield curve) \\
                \midrule
                $\texttt{normal}$ & $\beta_2 \ge 0,\ \beta_1 \ge \beta_2$ & $\beta_1 \le -\abs{\beta_2}$\\
                $\texttt{inverse}$ & $\beta_2 \le 0,\ \beta_1 \le \beta_2$ & $\beta_1 \ge \abs{\beta_2}$\\
                $\texttt{humped}$& $\beta_2 > 0,\ \beta_1 < \beta_2$ & $\beta_2 > \abs{\beta_1}$\\
                $\texttt{dipped}$ & $\beta_2 < 0,\ \beta_1 > \beta_2$ & $\beta_2 < -\abs{\beta_1}$\\
                \bottomrule
            \end{tabular}
            \caption{Parameter space segmentation for the Nelson-Siegel family}
            \label{tab:Nelson-Siegel}
        \end{table}

\subsection{Basic results on the Svensson family}\label{sub:TP}
For the Svensson family, the direct approach used for the Nelson-Siegel family runs into problems. Due to the additional basis function it is no longer possible to determine the zeroes of $\phi_S'$ explicitly, as it was in the Nelson-Siegel model. Nevertheless, some basic results can be obtained from the theory of Tchebycheff systems; see \cite{KS66, Kar68} for general results and \cite{KR21} for an application to interest rate models. From \cite{KS66} we recall that a family $(u_1, \dotsc, u_n)$ of continuous functions on a closed interval $[a,b]$ is a Tchebycheff system if 
\begin{equation}\label{eq:Phi}
\det \begin{pmatrix} u_{1}(x_1) & u_{2}(x_1) & \dotsc &  u_{n}(x_1)\\ \vdots & \vdots && \vdots\\  u_{1}(x_n) &  u_{2}(x_n) & \dotsc &  u_{n}(x_n)\end{pmatrix}  > 0
\end{equation}
holds for all $a \le x_1 < x_2 < \dotsm < x_n \le b$. When the family is defined on $[0,\infty)$ it is called a Tchebycheff system if it is a Tchebycheff system on $[0,A]$ for every $A > 0$. The key property of Tchebycheff systems is the following:
\begin{lemma}\label{lem:T}Let $(u_1, \dotsc, u_n)$ be a Tchebycheff system on $[0,\infty)$ and let $\bm{\beta} = (\beta_1, \dots, \beta_n) \in \RR^n \setminus \set{0}$. Then the function 
\[u(x) = \beta_1 u_1(x) + \dotsm \beta_n u_n(x)\]
has at most $n-1$ zeroes in $[0,\infty)$. 
\end{lemma}
Differentiating \eqref{eq:svensson} with respect to $x$, we find that 
        \begin{equation}\label{eq:svensson_diff}
            \phi_{S}'(x) = (\beta_2 - \beta_1) f_1(x) - \beta_2 f_2(x) + \beta_3 f_3(x) - \beta_3 f_4(x),
        \end{equation}
where
    \begin{equation*}
\Big(f_1(x), f_2(x), f_3(x), f_4(x)\Big)=\left(\tfrac{1}{\tau_1}e^{-\frac{x}{\tau_1}}, \tfrac{x}{\tau_1^2} e^{-\frac{x}{\tau_1}}, \tfrac{1}{\tau_2}e^{-\frac{x}{\tau_2}}, \tfrac{x}{\tau_2^2}e^{-\frac{x}{\tau_2}}\right).
\end{equation*}
It is not difficult to show that these basis functions $(f_1, f_2, f_3, f_4)$ of the differentiated forward curve in the Svensson family (as well as their counterparts for the yield curve) form a Tchebycheff system; see \cref{lem:tscheb}. This immediately leads to a simple proof of \cref{thm:basic}:
\begin{proof}[Proof of \cref{thm:basic}]
The derivative $\phi'_S$ of the Svensson curve $\phi_S$ is a linear combination of the functions $(f_1, f_2, f_3, f_4)$. The derivative of the Bliss curve ($\beta_2 = 0)$ is a linear combination of $(f_1, f_3, f_4)$. The derivative of the Nelson-Siegel curve ($\beta_3 = 0$) is a linear combination of $(f_1, f_2)$. By~\cref{lem:tscheb} $(f_1, f_2, f_3, f_4)$, $(f_1, f_3, f_4)$, and $(f_1, f_2)$ are Tchebycheff systems. Hence, in case of the forward curve, the claim follows by applying \cref{lem:T}. For the yield curve, it then suffices to apply \cref{thm:fw_to_yield}.
\end{proof}

\section{The envelope method}\label{sec:envelope}
We review the \emph{envelope method} for the classification of term structure shapes (\textbf{Q1}) and the segmentation of the parameter space $\Theta'$ (\textbf{Q2}). This method was introduced in \cite{KRS23} in the context of the two-factor Vasicek model, but can be applied in far more general settings. 

\subsection{Enveloping a family of lines}
Assuming for now that $\beta_3 \neq 0$ we differentiate \eqref{eq:svensson_2} and write it in the form
       \begin{equation}\label{eq:abc}
            \phi_S'(x) = \beta_3 \Big\{a_{\rm f}(x) + b_{\rm f}(x)\gamma_I + c_{\rm f}(x)\gamma_{II}\Big\}
                   \end{equation}
 with coefficients given by 
 \begin{equation}\label{eq:abc_explicit}
         \begin{split}
            a_{\rm f}(x) &= f_3(x) - f_4(x) = \frac{1}{\tau_2}\left(1-\frac{x}{\tau_2}\right)\exp\left(-\frac{x}{\tau_2}\right)\\
            b_{\rm f}(x) &= f_1(x) - f_2(x) = \frac{1}{\tau_1}\left(1-\frac{x}{\tau_1}\right)\exp\left(-\frac{x}{\tau_1}\right)\\
            c_{\rm f}(x) &= - f_1(x) = -\frac{1}{\tau_1}\exp\left(-\frac{x}{\tau_1}\right).
        \end{split}
        \end{equation}
        Note that the same structure \eqref{eq:abc} holds for the yield curve, but with the different basis functions 
         \begin{equation}\label{eq:abc_explicit_y}
         \begin{split}
		a_{\rm y}(x)&= \frac{1}{x^2}\left((x+\tau_2)e^{-x/\tau_2} - \tau_2\right)+ \frac{1}{\tau_2}e^{-x/\tau_2}\\
		b_{\rm y}(x)&= \frac{1}{x^2}\left((x+\tau_1)e^{-x/\tau_1} - \tau_1\right)+ \frac{1}{\tau_1}e^{-x/\tau_1}\\
		c_{\rm y}(x)&=  \frac{1}{x^2}\left((x+\tau_1)e^{-x/\tau_1} - \tau_1\right).
	\end{split}
\end{equation}
We will now drop the subscripts $\rm f$ and $\rm y$ and simply use $(a, b, c)$ to denote either of $(a_{\rm f}, b_{\rm f}, c_{\rm f})$ or $(a_{\rm y}, b_{\rm y}, c_{\rm y})$. Later, when it becomes necessary to distinguish between the two cases, we add superscripts and subscripts again.
It is easy to verify
that $c(x) \neq 0$ for all $x \in (0,\infty)$; hence, for any $x \in (0,\infty)$, 
 \begin{equation}\label{eq:line}
            \ell_x = \{ (\gamma_I, \gamma_{II})\in \RR^2\ |\ a(x) + b(x)\gamma_I + c(x)\gamma_{II} = 0\}
\end{equation}
defines a line in the $\gamma$-plane. The relevance of this family is as follows: 
\begin{quote}
{\textit
If $\ell_x$ passes through the point $ \gamma  = (\gamma_I, \gamma_{II})$, then any yield curve or forward curve from the Svensson family with parameter $\beta = (\beta_0, \beta_3 \gamma_{II}, \beta_3 \gamma_I, \beta_3)$ has a local extremum at $x$.}
\end{quote}
The type of extremum (dip/hump) is determined by the sign of $\beta_3$ and by keeping track of the half-space 
\[            \ell^+_x = \{ \gamma = (\gamma_I, \gamma_{II})\in \RR^2\ |\ a(x) + b(x)\gamma_I + c(x)\gamma_{II} > 0\}\]
and whether $\gamma$ is located in $\ell_x^+$ immediately before or after it is swept by $\ell_x$. As described in \cite{KRS23}, a non-linear contour appears at the boundary of the region swept by the lines $\ell_x$; see \cref{fig:envelope}. This contour is called the \emph{envelope} of the family $\cF = (\ell_x)_{x > 0}$ and is a classic object of interest in the geometry of lines and curves in the plane, see \cite{BG92}. It is precisely this contour which, together with $\ell_0$ and $\ell_\infty$ partitions the parameter space $\Theta'$ into regions of different forward-curve shapes and hence solves the \emph{segmentation problem} \textbf{Q2}. Mathematically, the envelope is defined as follows:
\begin{defn}The \emph{envelope} $\eta$ consists of all points $\gamma = (\gamma_I, \gamma_{II}) \in \RR^2$ which simultaneously satisfy
\begin{align}\label{eq:envelope}
\begin{split}
a(x) + b(x)\gamma_\RN{1} + c(x)\gamma_\RN{2} &= 0, \quad \text{and}\\
a'(x) + b'(x)\gamma_\RN{1} + c'(x)\gamma_\RN{2}  &= 0.
\end{split}
\end{align}
\end{defn}
Under mild conditions, equations \eqref{eq:envelope} are satisfied by a single point $\eta(x) = (\gamma_\RN{1}, \gamma_{II})$  for each $x \in (0,\infty)$ and their totality $\eta = \eta(x)_{x=0}^\infty$ forms a continuously differentiable curve in $\RR^2$. 

\begin{figure}
        
            \centering
            \begin{subfigure}{0.49\textwidth}
            \includegraphics[width=\textwidth]{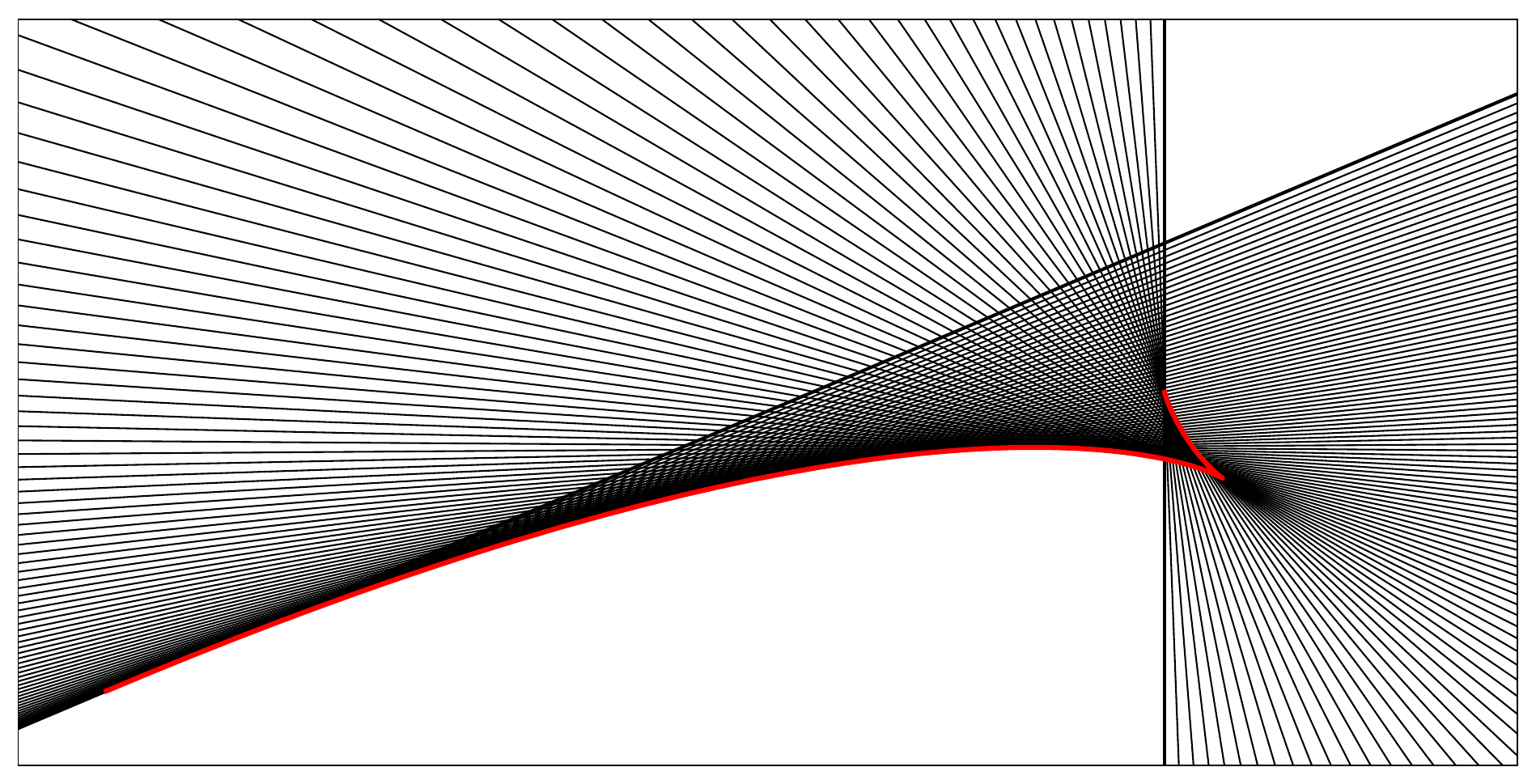}
   			\caption{}
            \label{fig:envelope}
            \end{subfigure}
            \begin{subfigure}{0.49\textwidth}
 	\includegraphics[width=\textwidth]{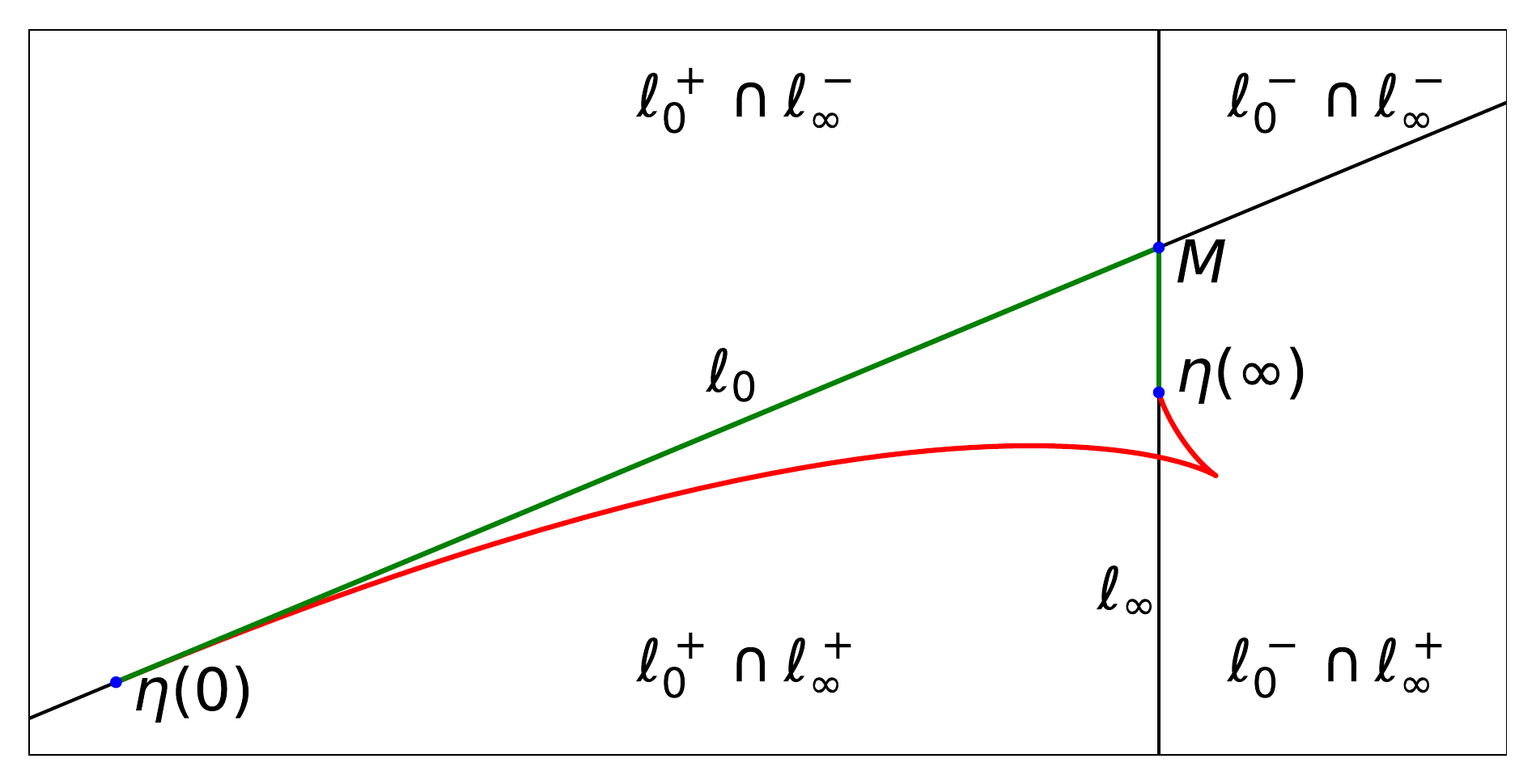}
   			\caption{}
            \label{fig:augmented_envelope}
            \end{subfigure}
            \caption{Left subplot: A family of lines and their envelope (red) appearing as non-linear contour. Right subplot: Schematic illustration of envelope $\eta$ (red), augmented envelope $\hat \eta$ (red + green), contact points $\eta(0)$, $\eta(\infty)$, lines $\ell_0$, $\ell_T$, intersection point $M$ and quadrants $\ell_0^\pm \cap \ell_\infty^\pm$.} \label{fig:envelope_schematic}
\end{figure}

\subsection{Assumptions for the envelope}\label{sub:assumptions}
We revisit some assumptions on the basis functions $a,b,c$ and their implications on the properties of the resulting envelope $\eta$, as introduced in \cite{KRS23}. Although these results were obtained in the context of the two-dimensional Vasicek model, they only use the structure \eqref{eq:abc} of the differentiated yield curve or forward curve, and are completely model-free. For reasons that will become apparent later, we consider a family of lines 
\[\cF_{(\alpha,\omega)} = \left(\ell_x\right)_{x \in (\alpha, \omega)}\]
indexed by a general open interval $(\alpha, \omega)$ with $0 \le \alpha < \omega \le \infty$. It should be obvious that all arguments of \cite{KRS23} for the case $\alpha = 0, \omega = \infty$ carry over to this slightly more general setting. 
The first assumption is 
    \begin{enumerate}[\bf ({A}1')]
        \item $c(x) \neq 0$ for all $x \in (\alpha,\omega)$.\label{item:nonzero}
    \end{enumerate}
 which makes sure that \eqref{eq:line} defines a non-degenerate line for each $x \in (\alpha,\omega)$.\footnote{\cite{KRS23} also require that $b(x) \neq 0$ for all $x \in (\alpha, \omega)$, but this part of the assumption is never used in later arguments.}
        The next assumption (slightly reformulated in comparison with \cite{KRS23}) guarantees well-defined limiting lines\footnote{Note that convergence of sets to a limiting set is equivalent to the pointwise convergence of their indicator functions.} $\ell_\alpha$ and $\ell_\omega$
            \begin{enumerate}[\bf ({A}1)]
                        \setcounter{enumi}{1}
     \item The lines $\ell_x$ converge to limiting lines $\ell_\alpha = \lim_{x \downarrow \alpha} \ell_x$ and $\ell_\omega = \lim_{x \uparrow \omega}\ell_x$. The same is true for the half-spaces $\ell^\pm_\alpha = \lim_{x \downarrow \alpha} \ell^\pm_\alpha$ and $\ell_\omega^\pm = \lim_{x \uparrow \omega} \ell_x^\pm$
     
     \label{item:limits}
        \end{enumerate}
 We also need the following non-degeneracy condition:
\begin{enumerate}[resume*]
\item The limiting lines $\ell_\alpha$ and $\ell_\omega$ are not parallel,\label{item:intersection}
\end{enumerate}
which guarantees that the lines $\ell_\alpha$ and $\ell_\omega$ have a unique intersection point $M = M_{(\alpha,\omega)}$.\\
Using the notation
\begin{equation}W(f_1, f_2, \dotsc, f_n)(x) = \det \begin{psmallmatrix} f_1(x) & f_2(x) & \dotsm & f_n(x)\\f'_1(x) & f'_2(x) & \dotsm & f'_n(x) \\ ... \\ f^{(n)}_1(x) & f^{(n)}_2(x) & \dotsm & f^{(n)}_n(x) \end{psmallmatrix}\end{equation}
for the \emph{Wronskian determinant} of a tuple of functions $f_1\dotsc, f_n$, we introduce the next assumption: 
\begin{enumerate}[resume*]
        \item $W(b,c)(x) \neq 0$ for all $x \in (\alpha, \omega)$.\label{item:Wbc}
    \end{enumerate}
This assumption guarantees a unique solution of the linear system \eqref{eq:envelope} for all $x \in (\alpha,\omega)$. From \cite{KRS23} we have the following result:\footnote{The last claim on the sign of $s'$ was shown in the proof, but not included in the statement of Lemma~3.1 in \cite{KRS23}.}
\begin{lemma}[Lemma~3.1 in \cite{KRS23}]\label{lem:envelope}
Under assumptions \ref{item:nonzero} and \ref{item:Wbc} the envelope $\eta = (\eta(x))_{x \in (\alpha,\omega)}$ of $\cF_{(\alpha,\omega)}$ is a continuously differentiable curve in $\RR^2$ given by 
        \begin{align}\label{eq:envelope_explicit}
        \eta(x) = \left(\frac{W(c,a)(x)}{W(b,c)(x)}, \frac{W(a,b)(x)}{W(b,c)(x)}\right)
    \end{align}
    and with tangent vector 
        \begin{align}
        \eta'(x) = \frac{W(a,b,c)(x)}{W(b,c)(x)^2} \begin{pmatrix}c(x)\\-b(x)\end{pmatrix}.
    \end{align}
    Moreover, the slope function $x \mapsto s(x) := -\frac{b(x)}{c(x)}$ of the lines $(\ell_x)_{x \in (0,\infty)}$ is strictly monotone; its derivative $s'$ has the same sign as $W(b,c)$.
    \end{lemma}
    We see that the tangent vector $\eta'(x)$ of the envelope vanishes if and only if $W(a,b,c) = 0$ or, equivalently, when
    \begin{equation}\label{eq:regression}
a''(x) + b''(x)\gamma_I + c''(x)\gamma_{II}= 0.
    \end{equation}
    Points of the envelope where \eqref{eq:regression} holds in addition to \eqref{eq:envelope}, are called \emph{points of regression}. We assume
        \begin{enumerate}[resume*]
        \item $W(a,b,c)(x) = 0$ for only finitely many $x \in (\alpha, \omega)$,\label{item:Wabc}
\newcounter{saveenum}
  \setcounter{saveenum}{\value{enumi}}
    \end{enumerate}
which guarantees that there are only finitely many such points. To make the connection to the envelope's shape, recall that a point of a differentiable plane curve is called \emph{singular} if the tangent vector vanishes, and such a point is called a \emph{cusp}, if the tangent vector changes direction while passing through it. All these notions are connected by the following result:
    \begin{lemma}[Lemma~3.5 in \cite{KRS23}]
\label{lem:singular0}
The following are equivalent for a point $\gamma = \eta(x)$ on the envelope of $\cF_{(\alpha, \omega)}$:
\begin{enumerate}[(a)]
\item $W(a,b,c)(x) = 0$;
\item $\gamma$ is a point of regression;
\item $\gamma$ is a singular point of $\eta$;
\end{enumerate}
Moreover, $\gamma$ is a cusp point iff $x$ is a transversal zero\footnote{A zero of a scalar function is called transversal if the function changes sign at the zero.} of $W(a,b,c)$.
 \end{lemma}
 By Lemma~\ref{lem:envelope} the slope of the envelope is monotone, such that its large-scale features are mostly determined by its cusps. Our final assumption is a technical assumption used in several proofs in \cite{KRS23}. We first recall the following notion: 
  \begin{defn}
        A line $h$ is in \emph{oblique position} with respect to the family $\cF_{(\alpha,\omega)} = (\ell_x)_{x \in (\alpha,\omega)}$, if it is not parallel to $\ell_x$ for any $x \in [\alpha,\omega]$,
    \end{defn}
\noindent and assume:
     \begin{enumerate}[resume*]
        \item There exists a line $h$ in oblique position with respect to $\cF_{(\alpha,\omega)}$. \label{item:oblique}
    \end{enumerate}
Let us assume that \ref{item:nonzero}-\ref{item:oblique} are in force, such that the envelope $\eta$ is given by \eqref{eq:envelope_explicit}. If the endpoints $\eta(\alpha) = \lim_{x \to \alpha} \eta(x)$ and $\eta(\omega) = \lim_{x \to \omega} \eta(x)$ exist, they must lie on $\ell_\alpha$ and $\ell_\omega$ and are therefore called \emph{contact points} of $\eta$ in \cite{KRS23}.\footnote{Even if they do not exist as proper points in $\RR^2$, a notion of `asymptotic contact' of $\eta$ to $\ell_\alpha$ or $\ell_\omega$ can be introduced, see \cite{KRS23}; however, in the context of the Svensson family, this will not be needed.} Augmenting $\eta$ with segments of $\ell_\alpha$ and $\ell_\omega$, it can be turned into a \emph{closed curve}:

\begin{defn}[cf. \cite{KRS23}]\label{defn:augmented}Let assumptions \ref{item:nonzero} - \ref{item:Wabc} hold, and let $M$ be the intersection point of $\ell_\alpha$ and $\ell_\omega$. The \emph{augmented envelope} $\hat \eta $ is the oriented curve with basepoint $M$, piecewise defined by 
\begin{enumerate}[(a)]
\item The line segment from $M$ to $\eta(\alpha)$ (contained in $\ell_\alpha$);
\item The envelope $\eta$; 
\item The line segment from $\eta(\omega)$ to  $M$ (contained in $\ell_\omega$).  
\end{enumerate}
\end{defn}
Adapted from \cite{KRS23} we have the following important result:
        \begin{thm}[See~Thm.~4.1 in \cite{KRS23}]\label{thm:winding}Let $E(\gamma)$ be the number of local extrema within the interval $(\alpha, \omega) \subseteq \Rplus$ of the forward curve or the yield curve with parameter $\gamma \in \Theta'$. Let $\hat \eta$ be the augmented envelope of the family $(\ell_x)_{x \in (\alpha, \omega)}$ associated to $f$ or $y$, and assume that \ref{item:nonzero}-\ref{item:oblique} holds true. Moreover, set
        \[D = \left(\ell_\alpha^+ \cap \ell_\omega^- \right) \cup \left(\ell_\alpha^- \cap \ell_\omega^+ \right) \]
        Then, for any $\gamma \in \RR^2 \setminus (\eta \cup \ell_\alpha \cup \ell_\omega)$
        \begin{equation}\label{eq:main_formula}
            E(\gamma) = 2 \abs{\wind_{\hat{\eta}}(\gamma)} + \mathbf{1}_{D}(\gamma),
        \end{equation}
        where $\wind_{\hat{\eta}}(\gamma)$ denotes the winding number of $\hat \eta$ around $\gamma$.
    \end{thm}
    
   As a corollary of this result, we obtain that the envelope $\eta$ together with the lines $\ell_\alpha$ and $\ell_\omega$ splits the parameter space $\Theta'$ into finitely many regions $R_1, \dotsc, R_k$ of constant term structure shape. The shape for a particular region $R_i$ can be determined from formula \eqref{eq:main_formula} and the initial sign of $\phi'_S$, resp.\ $\Phi'_S$. A schematic illustration is given in \cref{fig:augmented_envelope}. 
Finally, we remark that the proof of \cref{thm:winding} in \cite{KRS23} uses only the affine-linear form \eqref{eq:abc} of the differentiated term structure curves and assumptions \ref{item:nonzero}-\ref{item:oblique}, but is otherwise `model-free'.     

\section{Solving the classification and segmentation problem for the Svensson family}\label{sec:proofs}
In light of Theorem~\ref{thm:winding}, we have to check assumptions  \ref{item:nonzero}-\ref{item:oblique} in the context of the Svensson family, analyze the properties of the resulting envelope $\eta$, and determine the regions $R_1, \dotsc, R_k$ corresponding to the different term structure shapes. Together, this will allow us to show our main result \cref{thm:main} on the classification of term structure shapes and to obtain a clear understanding of the parameter space segmentation in the three different regimes described in the introduction. We start with a simple observation:
\begin{lemma}
The functions $c_{\rm f}$ and $c_{\rm y}$ have constant negative sign; hence \ref{item:nonzero} is satisfied for both yield- and forward curve.
\end{lemma}
Next we turn to the properties of the Wronskian determinants that appear in \ref{item:Wbc}, \ref{item:Wabc}, and \cref{lem:envelope}. Most results follow from direct (but sometimes tedious) computation, with some additional arguments taken from Appendix~\ref{app}.

\subsection{The Wronskian determinants}
\begin{lemma}\label{lem:W2_explicit}
The first-order Wronskians for the forward curve are given by 
        \begin{subequations}\label{eq:W2_explicit}
        \begin{align}\label{eq:Wbc_explicit}
                            W(b_{\rm f},c_{\rm f})(x) &= -\frac{1}{\tau_1^3}\exp\left(-\frac{2x}{\tau_1}\right),\\
            W(c_{\rm f},a_{\rm f})(x) &= \frac{\beta_3}{\tau_1^2\tau_2^2}\exp\left(-x\left(\frac{1}{\tau_1} + \frac{1}{\tau_2}\right)\right)\left[(2\tau_1 - \tau_2) - \frac{x}{\tau_2}(\tau_1 - \tau_2)\right],\\
            W(a_{\rm f},b_{\rm f})(x) &= \frac{\beta_3}{\tau_1^3\tau_2^3}(\tau_1 - \tau_2)\exp\left(-x\left(\frac{1}{\tau_1} + \frac{1}{\tau_2}\right)\right)\left[2\tau_1\tau_2 - x(\tau_1 + \tau_2) + x^2\right].
            \end{align}
        \end{subequations}
In particular, $W(b_{\rm f},c_{\rm f})$ has constant negative sign for all $x \in (0,\infty)$, and thus \ref{item:Wbc} holds for $\cF^{\rm f}_{(0,\infty)}$. The same is true for $W(b_{\rm y},c_{\rm y})$ and $\cF^{\rm y}_{(0,\infty)}$.
\end{lemma}
\begin{proof}
Equations \eqref{eq:W2_explicit} are obtained by direct calculation and the sign property of $W(b_{\rm f},c_{\rm f})$ is obvious. Applying \cref{lem:fw_to_yield}(d), the claim on $W(b_{\rm y},c_{\rm y})$ follows.
\end{proof}
Using \cref{lem:W2_explicit}, the envelope $\eta$ of $\cF_{(0,\infty)}$ can be calculated from \eqref{eq:envelope_explicit}. To check assumption \ref{item:Wabc}, we also need the second-order Wronskian.
\begin{lemma}\label{lem:Wabcf}
The Wronskian $W(a_{\rm f}, b_{\rm f},c_{\rm f})$ is given by
                \begin{gather}\label{eq:Wabc_si}
                    W(a_{\rm f},b_{\rm f},c_{\rm f})(x) = \frac{1}{\tau_1^5\tau_2^3}\exp\left(-x\left(\frac{2}{\tau_1} + \frac{1}{\tau_2}\right)\right)\left[-(2\tau_1 - \tau_2)^2 + \tau_1^2 + \frac{x}{\tau_2}(\tau_1 - \tau_2)^2\right].
                \end{gather}
It has no zero in $(0,\infty)$ if $r \in [1/3,1)$ and it has a single zero at
\begin{equation}\label{eq:cusp_time_fw}
x_*^{\rm f} = \frac{\tau_2(3\tau_1 - \tau_2)}{\tau_1 - \tau_2}.
\end{equation}
if $r \in (0,1/3) \cup (1,\infty)$; in particular \ref{item:Wabc} holds for $\cF^{\rm f}_{(0,\infty)}$. 
\end{lemma}
\begin{proof}
Again, both \eqref{eq:Wabc_si} and its zero $x_*^{\rm f}$ can be calculated directly, making the claim obvious.
\end{proof}

\begin{lemma}\label{lem:Wabcy}
The Wronskian $W(a_{\rm y},b_{\rm y},c_{\rm y})$ is given by 
                \begin{multline}\label{eq:large_Wronskian_y}
                    W(a_{\rm y},b_{\rm y},c_{\rm y})(x) = 
                    \frac{1}{x^4\tau_1^3\tau_2^3}\exp\left(-\frac{x}{\tau_1}\right)\cdot \\
                    \left[\exp\left(-x\left(\frac{1}{\tau_1} + \frac{1}{\tau_2}\right)\right)p_2(x) + \exp\left(-\frac{x}{\tau_2}\right)q_2(x) + \exp\left(-\frac{x}{\tau_1}\right)\tau_2^4\right] 
                \end{multline}
                where 
                \begin{align*}
                    p_2(x) &:= x^2(\tau_1^2 - 2\tau_1\tau_2 - \tau_2^2) + x(-2\tau_1^3 + 3\tau_1^2\tau_2 - 2\tau_1\tau_2^2 - \tau_2^3)\\ &\quad + \tau_2(4\tau_1^3 - 3\tau_1^2\tau_2 - \tau_2^3)\\
                    q_2(x) &:= x^2(\tau_1 - \tau_2) + x(\tau_2^2 + \tau_1\tau_2 - 2\tau_1^2) + \tau_1\tau_2(4\tau_1 - 3\tau_2).
                \end{align*}
                If $r \in [1/3,1)$ it has no zero in $(0,\infty)$ and if $r \in (0,1/3) \cup (1,\infty)$ it has a single zero $x^{\rm y}_*$, which must satisfy
                \[x_*^{\rm y} > x_*^{\rm f}.\]
                In particular \ref{item:Wabc} holds for $\cF^{\rm y}_{(0,\infty)}$. 
\end{lemma}
\begin{proof}
Let us write $W_{\rm y}(x) = W(a_{\rm y},b_{\rm y},c_{\rm y})(x)$ for short; the same is done for $W_{\rm f}(x)$. Equation \eqref{eq:large_Wronskian_y} follows by direct calculation. By \cref{lem:fw_to_yield}(e), $W_{\rm y}$ cannot have more zeroes than $W_{\rm f}$. From \cref{lem:fw_to_yield}(f), it follows that the initial sign of $W_{\rm y}$ is the same as for $W_{\rm f}$, that is, strictly positive if $r \in [1/3,1)$ and negative otherwise. The terminal sign, on the other hand, is always positive. We conclude that $W_{\rm y}$ has a single zero $x_*^{\rm y}$ if and only if $W_{\rm f}$ has one. Applying \cref{lem:fw_to_yield}(e) on the interval $(0,x_*^{\rm f})$ it follows that $x_*^{\rm y} > x_*^{\rm f}$.
\end{proof}

As a consequence of \cref{lem:singular0,lem:Wabcf,lem:Wabcy}, we see that the envelope $\eta$ has no singular point if $r \in [1/3,1)$ and has a single cusp point at $x_*$ if $r \in (0,1/3) \cup (1,\infty)$, explaining the phase transitions between the scale-regular, weakly scale-inverted and strongly scale-inverted regimes from \cref{def:regime}. Next, we focus on the analytic description of the parameter space segmentation. We start by looking at the limiting behaviour of $\cF_{(0,\infty)}$ at zero and infinity.

\begin{lemma} \ref{item:limits} holds for $\cF^{\rm f}_{(0,\infty)}$ and for $\cF^{\rm y}_{(0,\infty)}$ at the limit $x \to 0$. The limiting line is the same for the yield- and the forward curve and given by
\begin{equation}\label{eq:l0}
\ell_0: \quad \tau_1 + \tau_2 \gamma_I - \tau_2 \gamma_{II} = 0.
\end{equation}
The envelope contacts $\ell_0$ at the point 
        \begin{align}\label{eq:contact_zero}
            \eta(0) =\left(\frac{\tau_1}{\tau_2^2}(\tau_2 - 2\tau_1), \frac{2\tau_1}{\tau_2^2}(\tau_2 - \tau_1)\right).
        \end{align}
\end{lemma}
\begin{proof}
Also this Lemma follows by direct calculation from \eqref{eq:abc_explicit}.
\end{proof}

\begin{lemma} In the scale-regular case $r > 1$, \ref{item:limits} holds for $\cF^{\rm f}_{(0,\infty)}$ and for $\cF^{\rm y}_{(0,\infty)}$ at the limit $x \to \infty$. The limiting lines are given by
          \begin{align}\label{eq:l_inf}
                    \ell_\infty^{\rm f}:&\quad \gamma_I = 0, \quad \text{and}\\
                    \ell_\infty^{\rm y}:&\quad -\tau_2 - \tau_1\gamma_I - \tau_1\gamma_{II} = 0\notag
                \end{align}
for the forward curve and the yield curve respectively. These lines are not parallel to $\ell_0$, hence \ref{item:intersection} holds and the intersection points are given by
\begin{equation}\label{eq:M}
M_{\rm f} = \left(0,\frac{\tau_1}{\tau_2}\right), \qquad \text{resp.} \qquad M_{\rm y}= \left(-\frac{\tau_1^2 + \tau_2^2}{2 \tau_1 \tau_2}, \frac{\tau_1^2 - \tau_2^2}{2 \tau_1 \tau_2}\right).
\end{equation}
Moreover, an oblique line exists, and hence also \ref{item:oblique} is satisfied.  The envelope $\eta$ contacts $\ell_\infty$ at
\begin{equation}\label{eq:contact_inf}
\eta_{\rm f}(\infty) = (0,0) \qquad \text{resp.} \qquad \eta_{\rm y}(\infty) = \left(0,-\frac{\tau_2}{\tau_1}\right).
\end{equation}
\end{lemma}
\begin{proof}
Dividing the equation $a(x) + \gamma_I b(x) + \gamma_{II} c(x) = 0$ of $\ell_x$ by $b(x)$ and evaluating the limit $x \to \infty$ the equations for $\ell_\infty$ are obtained. It is obvious that these lines are not parallel to $\ell_0$; their intersection point is easily calculated. Finally, note that the slope of $\ell_0$ is $1$ and the (asymptotic) slopes of $\ell_\infty^{\rm f}$ and $\ell_\infty^{\rm y}$ are $-\infty$ and $-1$. By \cref{lem:envelope} the slopes $s(x)$ of $\ell_x$ are strictly decreasing. Hence, any line $h$ with slope strictly greater than one will be oblique to $\cF_{(0,\infty)}$ for both yield- and forward-curve.
\end{proof}

Finally, we look at the limit $x \to \infty$ in the case $r < 1$.
\begin{lemma} \label{lem:ell_scale_inverse}In the scale-inverted case  $\lim_{x \to \infty} \ell_x = \emptyset$ for both yield curve and forward curve. Hence, \ref{item:limits} \emph{does not} hold for the limit $x \to \infty$. The limits of the half-spaces are given by 
\begin{align*}
\lim_{x \to \infty} \ell_x^+  = \emptyset, \qquad \lim_{x \to \infty} \ell_x^- = \RR^2.
\end{align*}
\end{lemma}
\begin{proof}
The half-space $\ell_x^-$ consists of all $x \in \RR$ where 
\begin{equation}\label{eq:lplus}
a(x) + \gamma_I b(x) + \gamma_{II} c(x) < 0.
\end{equation} For large $x$, the value of $a(x)$ is negative. Dividing \eqref{eq:lplus} by $a(x)$ and taking into account that $b(x)/a(x) \to 0$ and $c(x)/a(x) \to 0$ we see that for any $\gamma = (\gamma_I, \gamma_{II}) \in \RR^2$ there is a large enough $x$, such that \eqref{eq:lplus} holds. Consequently $\lim_{x \to \infty} \ell_x^- = \RR^2$, and the claim follows.
\end{proof}
This shows that the cases $r < 1$ and $r > 1$ have to be treated differently: If $r > 1$, then \cref{thm:winding} can be applied directly to the family $\cF_{(0,\infty)}$ without further modifications. If $r < 1$, the method has to be adapted. Instead of $\cF_{(0,\infty)}$, we analyze the families  $\cF_{(0,T)}$ for large enough $T >0$ to classify the shape that the term structure takes on the bounded interval $[0,T]$. Then, exhausting $[0,\infty)$ by the intervals $[0,T]$, we derive results for term structure shapes over all of $[0,\infty)$. For applying \cref{thm:winding} to $\cF_{(0,T)}$ it only remains to show that \ref{item:intersection} and \ref{item:oblique} hold.
\begin{lemma}
	In the scale-inverted case, any line with slope strictly greater than one is oblique with respect to $\cF_{(0,T)}$ for all $T > 0$. Moreover, the line $\ell_T$ is never parallel to $\ell_0$; hence \ref{item:intersection} and \ref{item:oblique} hold true for any $T > 0$.
\end{lemma}
\begin{proof}
From \cref{lem:envelope} we know that the slope of $\ell_x$ is given by $s(x) = -\frac{b(x)}{c(x)}$. It is easily calculated that $s(0) = 1$ and that $\lim_{x \to \infty} s^{\rm f}(x) = -\infty$ and $\lim_{x \to \infty} s^{\rm y}(x) = -1$. By \cref{lem:envelope} $s(x)$ is strictly decreasing. Therefore, any line with slope strictly greater than one is oblique and the lines $\ell_x$ and $\ell_0$ can never be parallel. 

\end{proof}
We show one more Lemma, which is useful to derive the direction of the envelope in the neighborhood of its start- and endpoint.

    \begin{lemma}\label{lem:directions}
    Assume that \ref{item:nonzero}, \ref{item:limits}, \ref{item:Wbc}, and \ref{item:Wabc}
   	hold true for the family $\cF_{(\alpha, \omega)}$. Set
    \begin{align*}
    s_\alpha &= \lim_{x \downarrow \alpha} \sgn\left(W(a,b,c)(x)\,W(b,c)(x)\right), \\
     s_\omega &= \lim_{x \uparrow \omega} \sgn\left(W(a,b,c)(x)\,W(b,c)(x)\right).
    \end{align*}
    Then the half-space visited by the envelope immediately after its starting point $\eta(\alpha)$ is $\ell_\alpha^{s_\alpha}$, and the half-space visited by the envelope immediately before its endpoint $\eta(\omega)$ is $\ell_\omega^{s_\omega}$. 
    \end{lemma}
    \begin{proof}Consider the starting point $\eta(\alpha)$. The situation at the endpoint $\eta(\omega)$ is completely analogous. Let $h$ be small enough, such that $W(a,b,c)(x)$ has no zero for $x \in (\alpha, \alpha + h)$. The half-space visited by the envelope immediately after its starting point $\eta(\alpha)$ is determined by the sign of 
    \[m(h) = a(\alpha) + b(\alpha)\eta_1(\alpha+h) + c(\alpha)\eta_2(\alpha+h).\]
    Using the fact that $a(\alpha) + b(\alpha)\eta_1(\alpha) + c(\alpha)\eta_2(\alpha) = 0$, the mean-value theorem and \cref{lem:envelope}, there exists $h' \in (0, h)$ such that 
     \begin{align*}
            m(h) &= b(\alpha)\left(\eta_1(\alpha + h)  - \eta_1(\alpha) \right) + c(\alpha) \left(\eta_2(\alpha + h) -  \eta_2'(\alpha)\right)\\
            &= h \left(b(\alpha) \eta_1'(\alpha+h') + c(\alpha)\eta_2'(\alpha+h')\right)\\
            &= h \frac{W(a,b,c)(\alpha+h')}{(W(b,c)(\alpha+h'))^2} \left[b(\alpha)c(\alpha+h') - c(\alpha)b(\alpha+h')\right].
        \end{align*}
This expression has the same sign as $\tilde m(h') = \left(W(b,c)(\alpha+h')\right)^2 \frac{m(h)}{h}$, which after another application of the mean value theorem, is given by 
     \begin{align*}
            \tilde m(h') &= W(a,b,c)(\alpha+h') \left[b(\alpha)c(\alpha+h') - c(\alpha)b(\alpha+h')\right] = \\
            &= h'' W(a,b,c)(\alpha+h') \left[b(\alpha)c'(\alpha+h'') - c(\alpha)b'(\alpha+h'')\right]  = \\
            &= h'' W(a,b,c)(\alpha+h') W(b,c)(\alpha + h'')
            \end{align*}
for some $h'' \in (0,h')$, showing the Lemma. 
    \end{proof}

\subsection{Parameter-space segmentation in the scale-regular case}
We summarize the situation in the scale-regular case: The lines $\ell_0$ and $\ell_\infty$ are non-degenerate and given by \eqref{eq:l0} and \eqref{eq:l_inf}. The envelope starts at $\eta(0)$ given by \eqref{eq:contact_zero}, it has a single cusp point and ends at $\eta(\infty)$ given by \eqref{eq:contact_inf}.  Following \cite{KRS23} we relabel the `quadrants' resulting from the intersection of $\ell_0^\pm$ and $\ell_\infty^\pm$ as follows:
        \begin{defn}\label{def:quadrants1}Set $s = \sgn(\beta_3) ,\beta_3 \neq 0$. In the scale-regular case $r > 1$ we define the `quadrants'
        \begin{align*}
        Q_{n} &= {\ell_0^s}\cap {\ell_\infty^s} \qquad &Q_{h} &= \ell_0^s\cap\ell_\infty^{-s}\\
        Q_{i} &= {\ell_0^{-s}}\cap {\ell_\infty^{-s}} \qquad &Q_{d} &= \ell_0^{-s}\cap\ell_\infty^{s}.
 \end{align*}
 and their bounding rays
 \[r_{nd} = \partial Q_n \cap \partial Q_d, \quad  r_{id} = \partial Q_i \cap \partial Q_d, \quad r_{nh} = \partial Q_n \cap \partial Q_h, \quad r_{ih} = \partial Q_i \cap \partial Q_h.\]
        \end{defn}
        The mnemonics indicate that $Q_n$ contains the regions corresponding to \texttt{normal} curves; $Q_h$ to \texttt{humped} curves; $Q_i$ to \texttt{inverse} curves; $Q_d$ to \texttt{dipped} curves. The regions of `higher-order' shapes are determined by the envelope, see \cref{thm:winding}. Using these definitions, \cref{lem:directions}, and distinguishing the two possible signs of $\beta_3$, we conclude that in the scale-regular regime, there are two possible cases for the parameter-space segmentation:
        \begin{enumerate}
            \item $\beta_3 > 0$. In this case the envelope starts on $r_{nd}$ and moves into $Q_n$. It has an intersection with $\ell_\infty^{\rm f}$ and a cusp in $Q_h$. It has no intersections with $\ell_0$. The envelope ends on $r_{nh}$ coming from $Q_h$. The attainable shapes are \texttt{normal}, \texttt{inverse}, \texttt{humped}, \texttt{dipped}, \texttt{hd} and \texttt{hdh}. An example of the decomposition is shown in \cref{fig:b3g0} for the forward case and in \cref{fig:b3g0_y} for the yield case.
            \item $\beta_3 < 0$. The envelope and the half-spaces $\ell_0^\pm$ and $\ell_\infty^\pm$ are the same as for $\beta_3 > 0$; however all curve shapes (and the labelling of quadrants) have to be changed according to \cref{rem:polarity}, turning humps into dips and vice versa.
        \end{enumerate}

    \subsection{Parameter space segmentation in the scale-inverted case}
 Because of \cref{lem:ell_scale_inverse}, there will be only two non-empty `quadrants' resulting from the intersections of $\ell_0^\pm$ and $\ell_T^\pm$ in the limit $T \to \infty$, given by 

		\begin{align*}
			Q_{i} &= \ell_0^-, \qquad Q_{h} = \ell_0^+, \qquad \text{if $\beta_3 > 0$, and }\\
			Q_{d} &= \ell_0^+, \qquad Q_{n} = \ell_0^-, \qquad \text{if $\beta_3 < 0$}.
		\end{align*}
	We summarize the situation for the weakly scale-inverted and strongly scale-inverted case with $\beta_3 > 0$ here. The case $\beta_3 < 0$ can be handled similarly, see \cref{rem:polarity}. All attainable shapes, as well as some characteristics are given in \cref{tab:regimes}.
	\begin{enumerate}
		\item The weakly scale-inverted case with $\beta_3 > 0$. The envelope (for $x$ from $0$ up to some finite $T>0$) does not have any cusps and therefore no intersections with $\ell_0$ or $\ell_T$ or self-intersections. It moves into $\ell_0^- \cap \ell_T^-$ from its starting point and stays in the same quadrant until reaching $\ell_T$. The maximal winding number of the augmented envelope is one, and from \cref{thm:winding} we conclude that that the forward curve and the yield curve attains up to two local extrema in the interval $(0,T)$. Using \cref{lem:ell_scale_inverse} and sending $T$ to infinity, it follows that the term structure has to be eventually decreasing, such that exactly the shapes \texttt{inverse, humped, dh} are attainable. 
		\item The strongly scale-inverted case with $\beta_3 > 0$. For large enough $T > 0$ the envelope now has a cusp point before reaching $\ell_T$. It moves into $\ell_0^+ \cap \ell_T^-$ from its starting point, changes direction at its cusp point and must (for $T > 0$ large enough) intersect $\ell_0$ before reaching its endpoint on $\ell_T$, coming from $\ell_T^+$. The maximal winding number of the augmented envelope is one, and from \cref{thm:winding} we conclude that that the forward curve and the yield curve attains up to three local extrema in the interval $(0,T)$. Using \cref{lem:ell_scale_inverse} and sending $T$ to infinity, it follows that the term structure has to be eventually decreasing, such that exactly the shapes \texttt{inverse, humped, dh, hdh} are attainable. 
	\end{enumerate}
In summary, our analysis of all possible cases completes the proof of \cref{thm:main}.

		\begin{table}
			\scriptsize
			\captionsetup{width=0.9 \textwidth}
				\begin{adjustbox}{center}
					\begin{tabular}{|cc||c|c|c|c||c|c|c||c|}
						\hline
						\multicolumn{2}{|c||}{Regime} \vspace{.3em}& \rotatebox{90}{$\eta$ starts on} & \rotatebox{90}{$\eta$ moves into} & \rotatebox{90}{$\eta$ ends on} & \rotatebox{90}{$\eta$ exits from} & \rotatebox{90}{\# cusps} & \rotatebox{90}{\# inters.\ w/ $\ell_0$} & \rotatebox{90}{\# inters.\ w/ $\ell_\infty$} & \rotatebox{90}{\parbox{2.5cm}{shapes}}\\
						\hline
						\hline
						sr & $\beta_3 > 0$ & $r_{nd}$ & $Q_n$ & $r_{nh}$ & $Q_h$ & 1 & 0 & 1 & \texttt{n}, \texttt{i}, \texttt{h}, \texttt{d}, \texttt{hd}, \texttt{hdh}\\
						& $\beta_3 < 0$ & $r_{ih}$ & $Q_i$ & $r_{id}$ & $Q_d$ & 1 & 0 & 1 & \texttt{n}, \texttt{i}, \texttt{h}, \texttt{d}, \texttt{dh}, \texttt{dhd}\\
						\hline
						wsi & $\beta_3 > 0$ & $\ell_0$ & $Q_i$ & - & $Q_i$ & 0 & 0 & - & \texttt{i}, \texttt{h}, \texttt{dh}\\
						& $\beta_3 < 0$ & $\ell_0$ & $Q_n$ & - & $Q_n$ & 0 & 0 & - & \texttt{n}, \texttt{d}, \texttt{hd}\\
						\hline
						ssi & $\beta_3 > 0$ & $\ell_0$ & $Q_h$ & - & $Q_i$ & 1 & 1 & - & \texttt{i}, \texttt{h}, \texttt{dh}, \texttt{hdh}\\
						& $\beta_3 < 0$ & $\ell_0$ & $Q_d$ & - & $Q_n$ & 1 & 1 & - & \texttt{n}, \texttt{d}, \texttt{hd}, \texttt{dhd}\\
						\hline
					\end{tabular}
				\end{adjustbox}
		 		\caption{Regimes and attainable shapes of the Svensson family for $\beta_3 \neq 0$}
		 		\label{tab:regimes}
		\end{table}

		\begin{figure}
			\centering
			
						\begin{subfigure}{0.8\textwidth}

		\end{subfigure}

			\begin{subfigure}{0.8\textwidth}
			\includegraphics[width=\textwidth]{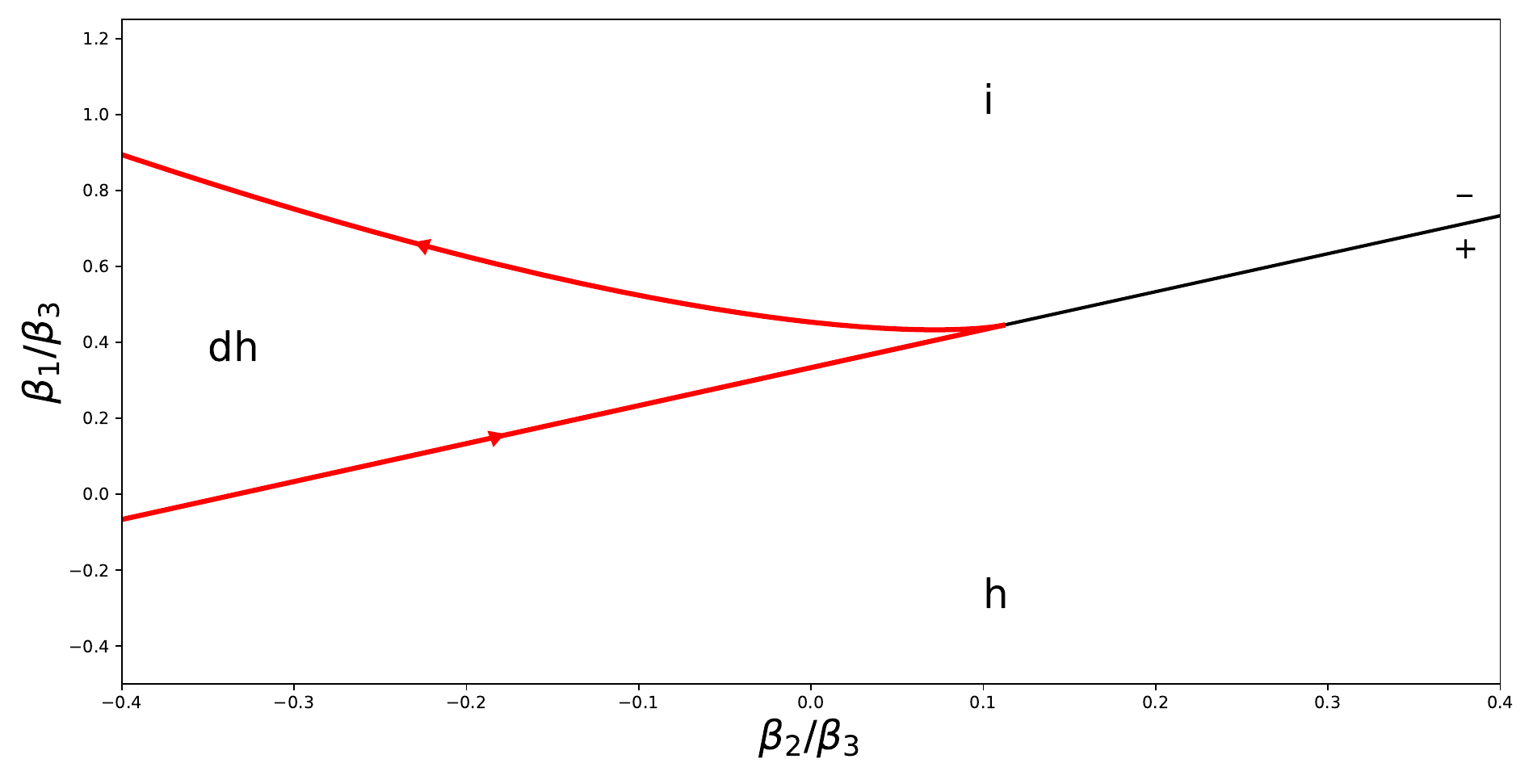}
			\caption{Forward Curve}
			\label{fig:t2e3t1_1}
			\end{subfigure}

			\begin{subfigure}{0.8\textwidth}
			\centering
			\includegraphics[width=\textwidth]{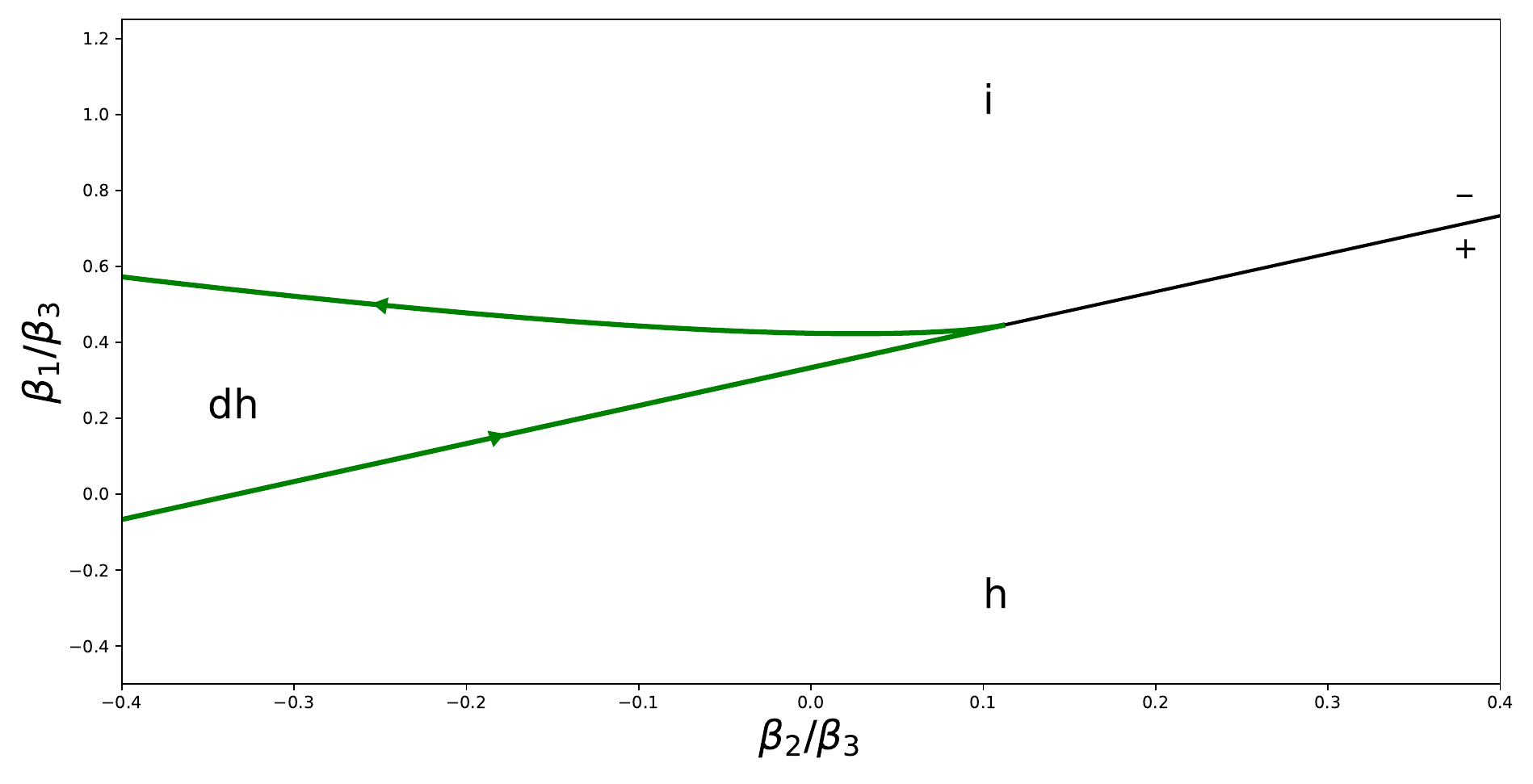}
			\caption{Yield Curve}
			\label{fig:t2e3t1_2}
			\end{subfigure}
			
			            \caption{Shapes of the forward curve (A) and the yield curve (B) for different regions of the parameter space $\Theta'$ of the Svensson family in the weakly scale-inverted regime $r \in [1/3,1)$. Red and green curves indicate the augmented envelope $\hat \eta$ (see Def.~\ref{defn:augmented}). Parameters used are $\tau_1 = 1, \tau_2 = 3$ and region labels (see \cref{tab:shape}) correspond to the case $\beta_3 > 0$. For $\beta_3 < 0$ the plot stays the same, but region labels must be changed, see \cref{rem:polarity}} \label{fig:wsi}
		
		\end{figure}
		
		\begin{figure}
			\centering

						\begin{subfigure}{0.8\textwidth}
			\includegraphics[width=\textwidth]{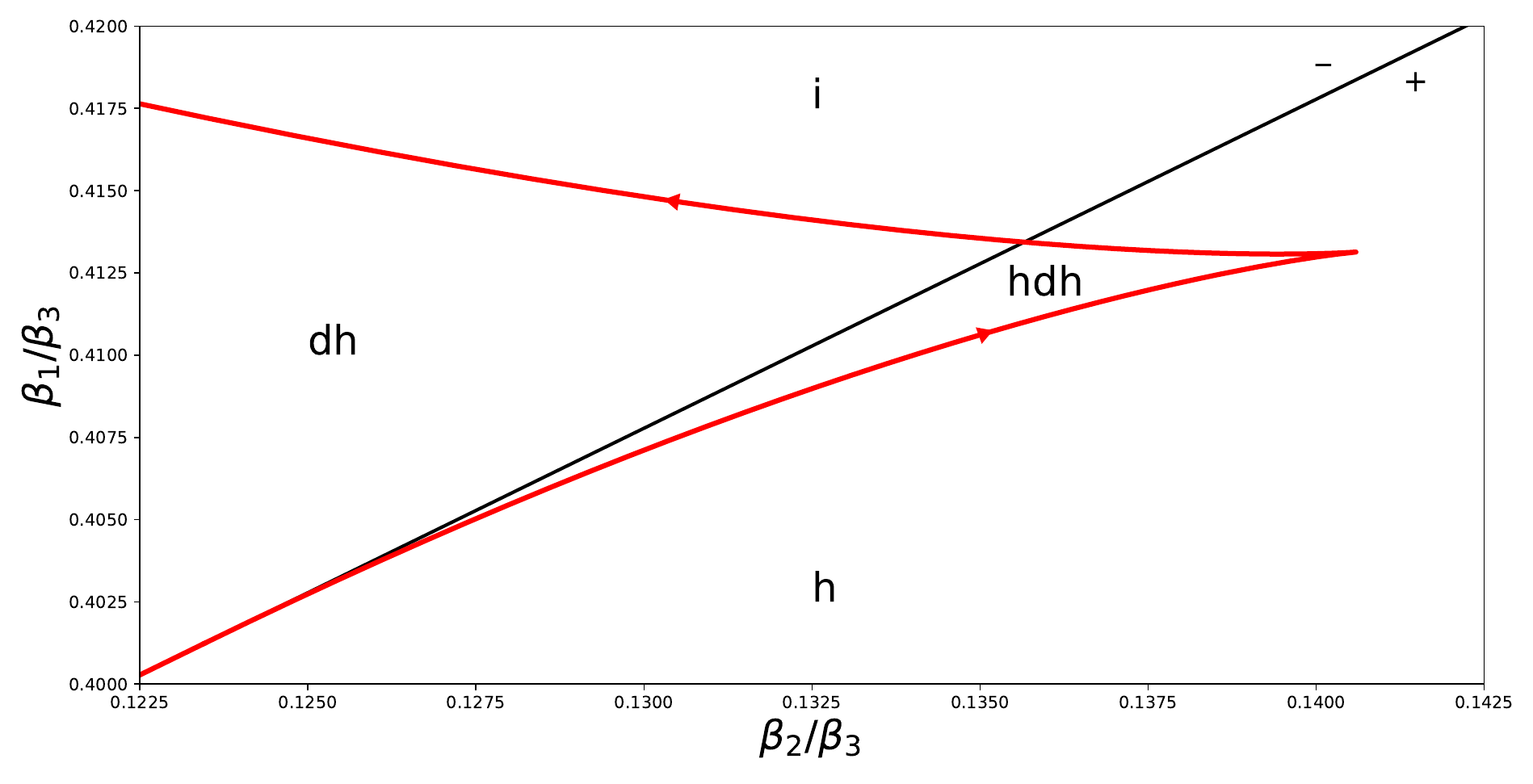}
			\caption{Forward Curve}
			\label{fig:t2e3.6t1_1}
		\end{subfigure}
		
					\begin{subfigure}{0.8\textwidth}
			\centering
			\includegraphics[width=\textwidth]{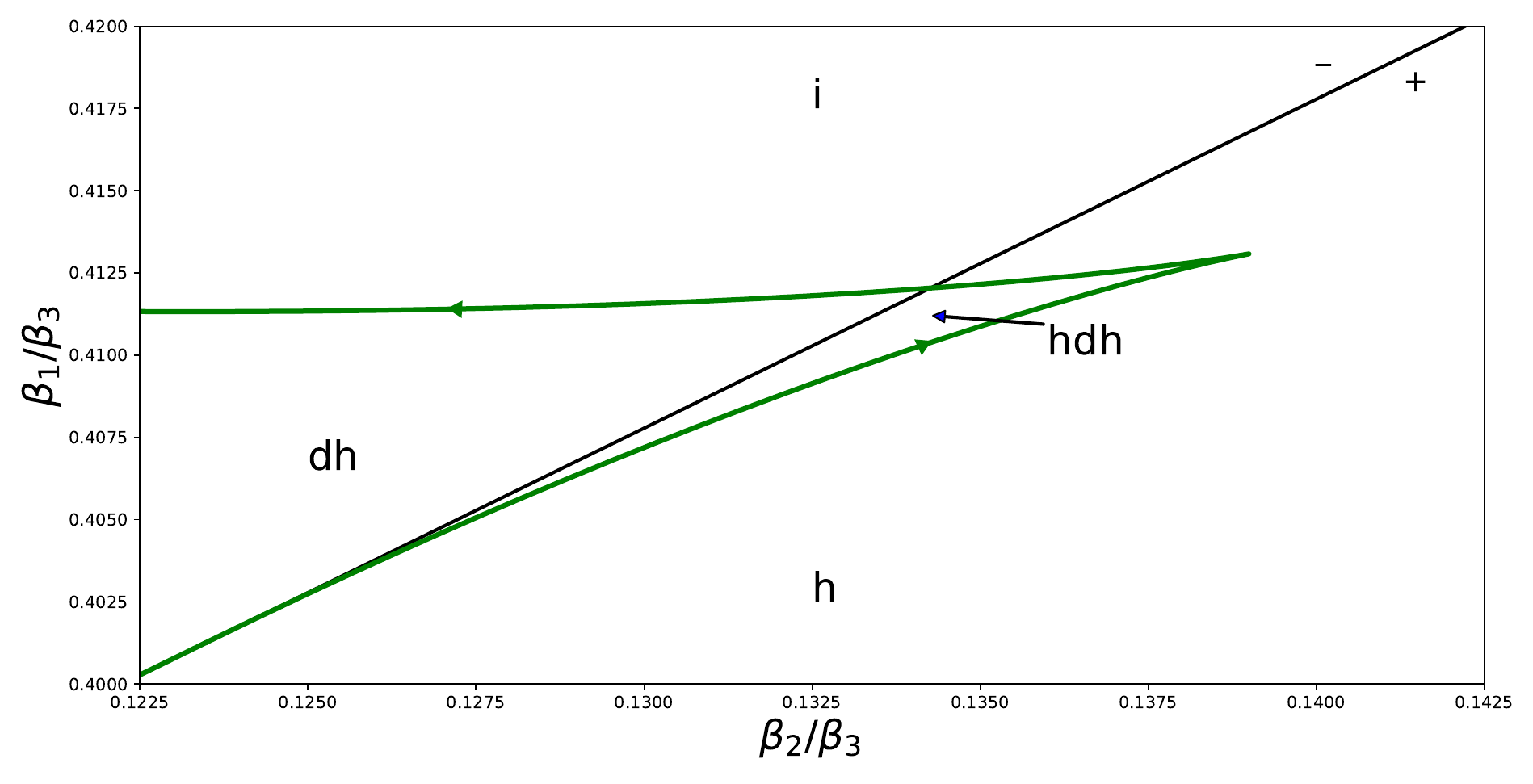}
			\caption{Yield Curve}
			\label{fig:t2e3.6t1_2}
		\end{subfigure}
		
			            \caption{Shapes of the forward curve (A) and the yield curve (B) for different regions of the parameter space $\Theta'$ of the Svensson family in the strongly scale-inverted regime $r \in (0,1/3)$. Red and green curves indicate the augmented envelope $\hat \eta$ (see Def.~\ref{defn:augmented}). Parameters used are $\tau_1 = 1, \tau_2 = 3.6$ and region labels (see \cref{tab:shape}) correspond to the case $\beta_3 > 0$. For $\beta_3 < 0$ the plot stays the same, but region labels must be changed, see \cref{rem:polarity}} \label{fig:ssi}
		
		\end{figure}
		

\subsection{The Bliss family}\label{sub:Bliss}
The Bliss family is the subfamily of the Svensson family that is obtained by setting $\beta_2 = 0$, or equivalently, $\gamma_\RN{1} = 0$. Thus, to obtain the attainable shapes of the Bliss family, we simply have to restrict the parameter-space $\Theta' = \set{(\gamma_\RN{1}, \gamma_\RN{2}) \in \RR^2}$ to the vertical line $\gamma_\RN{1} = 0$. For the forward curve in the scale-regular regime, we find that this line coincides exactly with $\ell^{\rm f}_\infty$, see \eqref{eq:l_inf}. This raises the question of how to apply \cref{eq:main_formula} to points which are located exactly at $\ell_0, \ell_\infty$ or on the envelope $\eta$, i.e. at the boundary of one of the regions $R_1, \dotsc, R_k$ of the parameter space. This question is answered in \cite[Rem.~4.1]{KRS23} as follows: The term structure shape associated to a point $\gamma$ at the boundary between two (or more) regions, is the shape of the adjacent region with the \emph{least number} of local extrema. Applying this rule to the parameter space segmentations obtained above and also checking the sign of the $\gamma_I$-coordinate of the midpoint $M_{\rm y}$ and the endpoint of the envelope $\eta_{\rm y}(\infty)$ in the scale-regular case, and the starting point of the envelope $\eta(0)$ in the scale-inverted case, we conclude the following:

\begin{cor}\label{cor:Bliss}
	The following statements hold for both the yield curve and the forward curve in the Bliss family:
	\begin{enumerate}[(a)]
		\item In the scale-regular regime the shapes \texttt{normal, inverse, humped, hd} are attainable if $\beta_3 > 0$. The shapes \texttt{normal, inverse, dipped, dh} are attainable if $\beta_3 < 0$.
		\item For $r\in [1/2, 1)$ the shapes \texttt{inverse, humped} are attainable if $\beta_3 > 0$. The shapes \texttt{normal, dipped} are attainable if $\beta_3 < 0$.
		\item For $r\in (0, 1/2)$ the shapes \texttt{inverse, humped, dh} are attainable if $\beta_3 > 0$. The shapes \texttt{normal, dipped, hd} are attainable if $\beta_3 < 0$.
	\end{enumerate}
\end{cor}
\begin{rem}
At this point it may be interesting to revisit the empirical observations on yield curves and forward curves shown in \cref{fig:regimes} and \cref{fig:shapes}. Firstly the shapes with three local extrema, i.e. \texttt{hdh} and \texttt{dhd}, occur (very) rarely; only in recent years have \texttt{hdh} shaped curves been observed. This indicates that the Bliss family may also have provided an adequate fit of bond yields over most of the time of observation. Secondly \cref{fig:regimes} shows that there can be long periods of time where $\tau_1$ and $\tau_2$ are nearly identical; in the figure the ratio is close to 1. This induces a multicolinearity problem which has been discussed in \cite{Pooter2007Examining}; its effects on $\beta_2$ and $\beta_3$ can also be observed in the ECB data. Working with the Bliss family could solve this problem and simultaneously provide a better interpretability of the parameters when $\tau_1$ and $\tau_2$ are close, at the cost of losing some goodness-of-fit.
\end{rem}
			
\section{Analysis of the consistent dynamic Svensson family}\label{sec:consistency}
As outlined in \cref{sub:consistency_preview}, we now consider a dynamic version of the Svensson-model, where the extended parameter vector 
$\tilde \beta = (\beta_0, \beta_1, \beta_2, \beta_3, \tau_1, \tau_2)$ is replaced by an Ito process $(\tilde \beta(t))_{t \ge 0}$ of the form \eqref{eq:ito}. From \cite{Fil00} (see also \cite[Prop. 9.5]{Fil09}) the only consistent dynamics of this process are of the following form
    \begin{align}
        \beta_0(t) &= \beta_0 \notag\\
        \dif \beta_1(t) &= \left(\frac{\beta_2}{\tau_1}e^{-t/\tau_1} + \frac{2\beta_3}{\tau_1}e^{-2t/\tau_1} - \frac{1}{\tau_1}\beta_1(t)\right)\dif t + \frac{\sqrt{2\beta_3}}{\tau_1}e^{-t/\tau_1}\dif W^*(t) \label{eq:SDE}\\ 
        \beta_2(t) &= \beta_2\exp\left(-\frac{1}{\tau_1}t\right) \notag\\
        \beta_3(t) &= \beta_3\exp\left(-\frac{2}{\tau_1}t\right) \notag\\
        \tau_1(t) &= \tau_1 \notag\\
        \tau_2(t) &= \frac{\tau_1}{2}, \notag
    \end{align}
    where $W^*$ is a one-dimensional $\QQ$-Brownian motion obtained as a transformation of the original multivariate $\QQ$-Brownian motion $W$; see \cite[Ch.~9]{Fil09} for details. We first observe that $\beta_3$ must be greater than 0 for the dynamics of $\beta_1(t)$ to be well-defined. This is already a significant restriction, since it implies per \cref{thm:main} that the shapes \texttt{dh} and \texttt{dhd} are not attainable. As a next step we transform to $\gamma$-coordinates, see \eqref{eq:gamma_coord}, setting
           \[\left(\gamma_\RN{1}(t), \gamma_\RN{2}(t)\right) = \left(\frac{\beta_2(t)}{\beta_3(t)}, \frac{\beta_1(t)}{\beta_3(t)}\right),\]
           which, by Ito's formula, implies
    \begin{align*}           
    \gamma_\RN{1}(t) &= \tfrac{\beta_2}{\beta_3} e^{t/\tau_1}\\
    \dif \gamma_\RN{2}(t) &= \frac{1}{\tau_1}\left(2 + \gamma_\RN{1}(t) + \gamma_\RN{2}(t)\right)\dif t + \frac{1}{\tau_1} \sqrt{\frac{2}{\beta_3(t)}} \dif W_t^* , \quad \gamma_\RN{2}(0) = \frac{\beta_1}{\beta_3}.
    \end{align*}
    In the parameter-space picture (see~\cref{fig:b3g0}), the coordinate $\gamma_\RN{1}(t)$ describes a monotone movement along the horizontal axis, which is increasing if $\beta_2 > 0$ and decreasing if $\beta_2 < 0$. The coordinate $\gamma_\RN{2}(t)$ follows a non-degenerate Gaussian process, which is supported on the whole real line for any $t > 0$. We conclude the following: 
    \begin{lemma}
    Under the consistent dynamic evolution of the Svensson family, a shape $\textsf{S}$ of the forward curve is attained with strictly positive probability at time $t > 0$, if and only if the vertical line defined by $\gamma_\RN{1}(t) = \frac{\beta_2}{\beta_3} e^{t/\tau_1}$ in the $(\gamma_\RN{1},\gamma_\RN{2})$-plane intersects the region $R_i$ associated to shape $\textsf{S}$.
    \end{lemma}
    To obtain \cref{thm:main_consistency}, it remains to determine the possible intersections. We proceed separately for the forward curve and the yield curve.

        \subsubsection{Case of the forward curve} 
        \begin{proof}[Proof of \cref{thm:main_consistency}]
         Inserting $\tau_2 = \tfrac{\tau_1}{2}$ into equation \eqref{eq:envelope_explicit} for the envelope, we get the expression
         \begin{equation}\label{eq:envelope_consistent}
         \eta^{\rm f}(x) = \left(\frac{4}{\tau_1} e^{-x/\tau_1}\left(x - \frac{3\tau_1}{2}\right), -\frac{4}{\tau_1^2} e^{-x/\tau_1}\left(x^2 - \frac{3x}{2}\tau_1 + \tau_1^2\right)\right).
         \end{equation}
         
Taking a look at the parameter space segmentation in the scale-regular regime (see \cref{fig:b3g0}), we see that the regions visited by the process $\gamma(t)$ at a given $t > 0$ depend only on the location of the cusp in the case $\beta_2 > 0$ and on the location of $\eta^{\rm f}(0)$ in the case $\beta_2 < 0$. From \eqref{eq:cusp_time_fw} we see that the cusp takes place at $x_*^{\rm f} = \frac{5 \tau_1}{2}$. From \eqref{eq:envelope_consistent} we determine the location of the cusp to be 
\[        \eta^{\rm f}(x_*^{\rm f}) = \left(4 e^{-5/2}, -14 e^{-5/2}\right).\]
Equating the first coordinate to $\gamma_\RN{1}(t)$ and solving for $t$ we find \eqref{eq:T_dagger_fw}. Turning to the case $\beta_2 < 0$, we do the same analysis for the contact point $\eta^{\rm f}(0)$ and find
\[        \eta^{\rm f}(0) = \left(-6,-4\right).\]
Equating the first coordinate to $\gamma_\RN{1}(t)$, we obtain \eqref{eq:T_dagger_fw}, completing the proof. 

\end{proof}
        \subsubsection{Case of the yield curve}
        For the yield curve, the analogue of \cref{thm:main_consistency} is the following:
        
                    \begin{thm}\label{thm:consistency_yield} In the consistent dynamic Svensson-model with initial curve given by \eqref{eq:svensson} the following holds:
                \begin{enumerate}[(1)]
                    \item If $\beta_2 > 0$ then there is a time-horizon $T^{\rm y}_\dagger \in [0,T^{\rm f}_\dagger)$ such that the yield curve attains each of the shapes \texttt{inverse}, \texttt{humped}, \texttt{hd} and \texttt{normal} with strictly positive probability at any time $t < T^{\rm y}_\dagger$, and each of the shapes \texttt{inverse}, \texttt{humped}, \texttt{normal} with strictly positive probability at any time $t \ge T^{\rm y}_\dagger$.
                    \item If $\beta_2 < 0$ then there are time-horizons $T_{\star \star}^{\rm y} \le T_{\star}^{\rm y}$, given by 
                    \[T_{\star \star}^{\rm y} = \max\left(\log\left(\frac{5 \beta_3}{4 |\beta_2|}\right),0\right), \quad T_\star^{\rm y} = \max\left(\log\left(\frac{6 \beta_3}{|\beta_2|}\right),0\right)\]
                    such that the yield curve attains each of the following shapes with strictly positive probability: 
                        \begin{enumerate}[(a)]
                            \item \texttt{inverse}, \texttt{humped}, \texttt{hd} and \texttt{normal} if $t < T_{\star \star}^{\rm y}$
                            \item \texttt{inverse}, \texttt{hd} and \texttt{normal} if $t =  T_{\star \star}^{\rm y}$
                            \item \texttt{inverse}, \texttt{dipped}, \texttt{hd} and \texttt{normal} if $T_{\star \star}^{\rm y} < t < T_\star^{\rm y}$
                            \item \texttt{inverse}, \texttt{dipped} and \texttt{normal} if $t \ge T_\star^{\rm y}$.
                        \end{enumerate}
                \end{enumerate}
            \end{thm}

\begin{proof}       Taking a look at the parameter space segmentation in the scale-regular regime (see \cref{fig:b3g0_y}), we see that the regions visited by the process $\gamma(t)$ at a given $t > 0$ depend only on the location of the cusp in the case $\beta_2 > 0$ and on the location of $\eta^{\rm y}(0)$ and $M_{\rm y}$ in the case $\beta_2 < 0$. Inserting $\tau_2 = \frac{\tau_1}{2}$ into \cref{eq:contact_zero,eq:contact_inf,eq:M}, we compute their coordinates as
\begin{equation*}
                \eta_{\rm y}(0) = (-6, -4),\qquad 
                \eta_{\rm y}(\infty) = (0,-\tfrac{1}{2}),\qquad 
                M_{\rm y} = (-\tfrac{5}{4}, \tfrac{3}{4}).
\end{equation*}
Let $\beta_2 < 0$. Setting the first coordinates of  $\eta_{\rm y}(0)$ and $M_{\rm y}$ equal to $\gamma_\RN{1}(t)$ and solving for $t$ we obtain the time-points $T^{\rm y}_\star$ and $T^{\rm y}_{\star \star}$. Now consider $\beta_2 > 0$. The cusp point of $\eta^{\rm y}$ is not known explicitly, but we know from \cref{lem:Wabcy} that it occurs at a later time $x_*^{\rm y} > x_*^{\rm f}$ than the cusp point of $\eta^{\rm f}$, and -- from \cref{lem:fw_to_yield}(c) -- that it must be incident to $\eta^{\rm f}$. Together, this allows us to conclude that the time $T_\dagger^{\rm y}$ where the vertical at $\gamma_\RN{1}(t)$ meets the cusp must satisfy $T_\dagger^{\rm y} < T_\dagger^{\rm f}$.
\end{proof}
                
    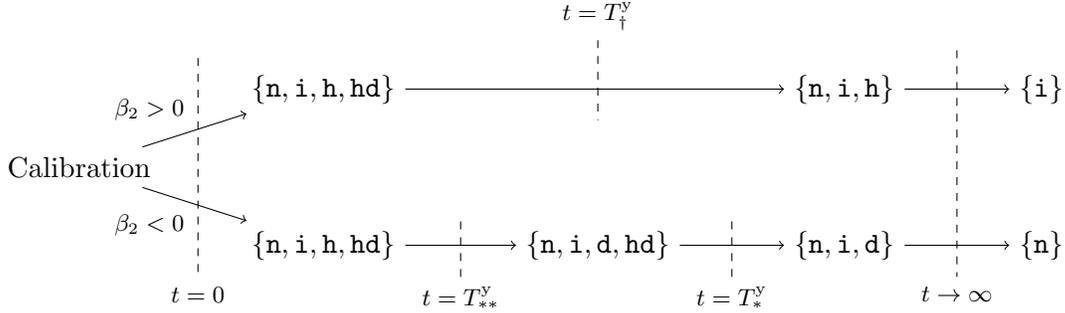
\begin{figure}

    	\begin{tikzpicture}\hspace{-3em}
  \matrix[row sep=1mm,column sep=0.1cm] {

     \node[right] (h1) {\footnotesize $\phantom{t=0}$}; & & & \node[above] (h2) {\footnotesize $t = T_\dagger^{\rm y}$}; &  & & \node (h3) {}; & \\
      & \node (n1)  {$\{\texttt{n}, \texttt{i},\texttt{h},\texttt{hd}\}$}; &  & &  &  \node (n2) {$\{\texttt{n}, \texttt{i},\texttt{h}\}$}; & & \node (n3)  {$\{\texttt{i}\}$};\\
 & & & \node (h2up) {}; & & & &  \\
    \node[left] (n0)  {Calibration}; &  & & & &  & & \\
 & & \node (h5down) {}; & & \node (h6down) {}; & & & \\    
        & \node (n4)  {$\{\texttt{n}, \texttt{i},\texttt{h},\texttt{hd}\}$};  &  & \node (n5) {$\{\texttt{n}, \texttt{i},\texttt{d},\texttt{hd}\}$};  &  & \node (n6) {$\{\texttt{n}, \texttt{i},\texttt{d}\}$}; & & \node (n7)  {$\{\texttt{n}\}$};\\
           \node[below right] (h4) {\footnotesize $t=0$}; & & \node[below] (h5) {\footnotesize $t = T_{**}^{\rm y}$}; & & \node[below] (h6) {\footnotesize $t = T_*^{\rm y}$}; & & \node[below] (h7) {\footnotesize $t \to \infty$}; & \\
  };
  \graph {
  (n0) ->["{\footnotesize $\beta_2 > 0$}"] (n1) -> (n2) -> (n3);
  (n0) ->["{\footnotesize $\beta_2 < 0$}" below left] (n4) -> (n5) -> (n6) -> (n7);
  (h1) --[dashed] (h4); 
  (h2) --[dashed] (h2up); 
  (h5) --[dashed] (h5down); 
  (h6) --[dashed] (h6down); 
  (h3) --[dashed] (h7); 
  };
\end{tikzpicture}

    	\caption{This diagram shows the shapes that are attained with strictly positive probability for the yield curve in the consistent Svensson model as time $t$ progresses, illustrating Thms.~\ref{thm:consistency_yield} and \ref{thm:ergodic}. See Figure~\ref{fig:shape-flow_chart_fw} for further comments.}
    	\label{fig:shape-flow_chart_y}
    \end{figure}

    \subsection{Exact shape probabilities}\label{sub:frequency}
    Our next goal is to compute the exact risk-neutral probabilities with which a given shape $\textsf{S}$ is observed at time $t > 0$. In light of our previous result, this problem reduces to computing the probability with which the process $\gamma_\RN{2}(t)$ hits the region corresponding to $\textsf{S}$. The process $\gamma_\RN{2}(t)$ is Gaussian, hence its distribution is completely characterized by its mean and variance, which are given as follows:
\begin{lemma}The SDE \eqref{eq:SDE} has the explicit solution
\begin{equation}\label{eq:gamma_explicit}
\gamma_\RN{2}(t) = e^{t/\tau_1} \left(\tfrac{\beta_2}{\beta_3 \tau_1} t + \frac{\beta_1}{\beta_3} + 2\right) - 2 + \frac{e^{t/\tau_1}}{\tau_1}\sqrt{\frac{2}{\beta_3}}W_t^*.
\end{equation}
The mean and variance of $\gamma_\RN{2}(t)$ are given by     
        \begin{equation}\label{eq:mu_sigma}
        	\mu_t = e^{t/\tau_1} \left(\tfrac{\beta_2}{\beta_3 \tau_1} t + \frac{\beta_1}{\beta_3} + 2\right) - 2
	\quad \text{and} \quad 
        	\sigma_t^2 = \frac{2t e^{2t/\tau_1}}{\beta_3 \tau_1^2} 
        \end{equation}
\end{lemma}
\begin{proof}
Setting $X_t = e^{-t/\tau_1} \gamma_\RN{2}(t)$ and applying the product rule for Ito processes, we see that
\[dX_t = \frac{1}{\tau_1} \left(2 e^{-s/\tau_1} + \frac{\beta_2}{\beta_3}\right) + \frac{1}{\tau_1} \sqrt{\frac{2}{\beta_3}}dW_t^*, \quad X_0 = \frac{\beta_1}{\beta_3}.\]
This SDE can be integrated, which leads to \eqref{eq:gamma_explicit}. The mean and the variance can now be read off from this equation.
\end{proof}        

To calculate the risk-neutral shape probabilities, we again proceed separately for forward- and yield-curve.
        \subsubsection{Case of the forward curve}
        	\begin{lemma}\label{lem:bounds_f}
        		Let $T_\dagger^{\rm f}$ and $T_\star^{\rm f}$ be given by \cref{eq:T_dagger_fw,eq:T_star_fw} and set 
        		\begin{align*}
        			I_0(t) &:= 2 + \frac{\beta_2}{\beta_3} e^{t/\tau_1}\\
        			x_0^I(t) &:= \tau_1\left(\frac{3}{2} - \mathcal{W}_0\left(\frac{\beta_2}{4\beta_3}\exp\left(\frac{t}{\tau_1} + \frac{3}{2}\right)\right)\right),\quad t \le T_\dagger^{\rm f}\\
        			x_{-1}^I(t) &:= \tau_1\left(\frac{3}{2} - \mathcal{W}_{-1}\left(\frac{\beta_2}{4\beta_3}\exp\left(\frac{t}{\tau_1} + \frac{3}{2}\right)\right)\right), \quad t \le T_\star^{\rm f}
        		\end{align*}
        		where $\mathcal{W}_0$ and $\mathcal{W}_{-1}$ denote the real branches of the Lambert $\mathcal{W}$ function (cf. \cite{CG96} and \cref{lem:Lambert}).
        		\begin{enumerate}
        			\item If $\beta_2 > 0$ and $t < T_\dagger$, then the forward curve is 
        			\begin{enumerate}
        				\item \texttt{inverse} if $\gamma_\RN{2}(t) \ge I_0(t)$
        				\item \texttt{humped} if $\eta^{\rm f}(x_{-1}^I(t)) \le \gamma_\RN{2}(t) < I_0(t)$ or $\gamma_\RN{2}(t) \le \eta^{\rm f}(x_{0}^I(t))$
        				\item \texttt{hdh} if $\eta^{\rm f}(x_{0}^I(t)) < \gamma_\RN{2}(t) < \eta^{\rm f}(x_{-1}^I(t))$
        			\end{enumerate}
        			\item If $\beta_2 > 0$ and $t \ge T_\dagger$, then the forward curve is 
        			\begin{enumerate}
        				\item \texttt{inverse} if $\gamma_\RN{2}(t) \ge I_0(t)$
        				\item \texttt{humped} if $\gamma_\RN{2}(t) < I_0(t)$
        			\end{enumerate}
        			\item If $\beta_2 < 0$ and $t < T_\star$, then the forward curve is 
        			\begin{enumerate}
        				\item \texttt{dipped} if $\gamma_\RN{2}(t) \ge I_0(t)$
        				\item \texttt{hd} if $\eta^{\rm f}(x_{0}^I(t)) < \gamma_\RN{2}(t) < I_0(t)$
        				\item \texttt{normal} if $\gamma_\RN{2}(t) \le \eta^{\rm f}(x_{0}^I(t))$
        			\end{enumerate}
        			\item If $\beta_2 < 0$ and $t \ge T_\star$, then the forward curve is 
        			\begin{enumerate}
        				\item \texttt{dipped} if $\gamma_\RN{2}(t) > I_0(t)$
        				\item \texttt{normal} if $\gamma_\RN{2}(t) \le I_0(t)$
        			\end{enumerate}
        		\end{enumerate}
		The risk-neutral probability of each shape can be calculated from the fact that $\gamma_\RN{2}(t)$ is normally distributed under $\mathbb{Q}$ with parameters given by \eqref{eq:mu_sigma}.
        	\end{lemma}
        	\begin{proof}
        		We obtain $I_0(t)$ by calculating the $\gamma_\RN{2}$-coordinate of the intersection of the line $\ell_0$ with the vertical line $v_t$ defined by $\gamma_\RN{1}(t) = \tfrac{\beta_2}{\beta_3}e^{t/\tau_1}$.  For the other bounds we have to find the intersection points of the envelope $\eta_{\rm f}$ with $v_t$, i.e. to solve
        		\begin{align*}
        			\gamma_\RN{1}(t) = \eta^{\rm f}(x).
        		\end{align*}
Writing the two sides in explicit form, this yields
        		\begin{align*}
        			\left(x - \frac{3}{2}\tau_1\right)\exp\left(-\frac{1}{\tau_1}x\right) = \frac{\tau_1\beta_2}{4\beta_3}\exp\left(\frac{1}{\tau_1}t\right).
        		\end{align*}
        		By applying \cref{lem:Lambert} we obtain the desired bounds.
         	\end{proof}
	
We are now prepared to show \cref{thm:ergodic} on the long-term behaviour of yield curves and forward curves under the consistent dynamics of the Svensson family.
	
         	\begin{proof}[Proof of \cref{thm:ergodic} for the forward curve] Let $\beta_2 > 0$. 
        Using \cref{lem:bounds_f} and \eqref{eq:gamma_explicit}, we have the equality, for any $\omega \in \Omega$, 
        \begin{equation}\label{eq:set_equality}
        \set{S_t(\omega) \text{ is } \texttt{inverse}} = \set{\gamma_\RN{2}(t, \omega) \ge I_0(t)} = \set{\frac{W_t^*(\omega) - \alpha_t}{t} \ge - \frac{\beta_2}{\sqrt{2 \beta_3}}},
        \end{equation}
        where $\alpha_t$ can be computed from \eqref{eq:gamma_explicit} as
        $\alpha_t = \tau_1 \left(\sqrt{2\beta_3}(2e^{-t/\tau_1} - 1) - \frac{\beta_1}{\sqrt{2 \beta_3}}\right)$.
        The only relevant property of $\alpha_t$ is that it is bounded. Setting $\epsilon =  \frac{\beta_2}{\sqrt{2 \beta_3}} > 0$ and using the strong law of large numbers for Brownian motion, there exists $A \in \cA$ with $\QQ(A) = 1$, such that, for every $\omega \in A$ there is a number $K(\omega) > 0$ such that
        \[\left| \frac{W_t^* - \alpha_t}{t}\right| \le \epsilon =  \frac{\beta_2}{\sqrt{2 \beta_3}} \qquad \text{for all $t \ge K(\omega)$.}\]
        Hence, we have
\[	\mathbf{1}\set{\mathsf{S}_t \text{ is } \texttt{inverse}}  = 1 \qquad \text{for all $t \ge K(\omega)$.}\]
Sending $t \to \infty$ this shows \eqref{eq:ergodic1}. In the case $\beta_2 < 0$, the argument is analogous, but with equality signs in \eqref{eq:set_equality} reversed. \end{proof}
        
        \subsubsection{The case of the yield curve}
        The analogous result to \cref{lem:bounds_f} is the following:
        
        	\begin{lemma}\label{lem:bounds_y}
        		Let
        		\begin{align*}
        			I_\infty(t) := -\frac{\beta_2}{\beta_3}e^{t/\tau_1} - \frac{1}{2}.
        		\end{align*}
        		\begin{enumerate}
        			\item If $\beta_2 > 0$ and $t \ge T_\dagger^{\rm y}$, then the yield curve is
        			\begin{enumerate}
        				\item \texttt{inverse} if $\gamma_\RN{2}(t) \ge I_0(t)$,
        				\item \texttt{humped} if $I_\infty(t) < \gamma_\RN{2}(t) < I_0(t)$,
        				\item \texttt{normal} if $\gamma_\RN{2}(t) \le I_\infty(t)$.
        			\end{enumerate}
        			\item If $\beta_2 < 0$ and $t \ge T_{\star \star}^{\rm y}$, then the yield curve is
        			\begin{enumerate}
        				\item \texttt{inverse} if $\gamma_\RN{2}(t) \ge I_\infty(t)$,
        				\item \texttt{dipped} if $I_0(t) < \gamma_\RN{2}(t) < I_\infty(t)$,
        				\item \texttt{normal} if $\gamma_\RN{2}(t) \le I_0(t)$.
        			\end{enumerate}
        		\end{enumerate}
        	\end{lemma}
        	\begin{proof}
        		Analogous to the proof of \cref{lem:bounds_f}.
        	\end{proof}
	Note that bounds for the shape probabilities before the time $T_\dagger^{\rm y}$ (resp.\ $T_{\star \star}^{\rm y}$) can be obtained by combining \cref{lem:bounds_f} with \cref{thm:fw_to_yield}. To conclude, we note that also the proof of \cref{thm:ergodic} for the yield curve is completely analogous to the given proof for the forward curve, such that it can be omitted.
	
	\appendix
	\section{Auxiliary results}\label{app}
	\begin{lemma}\label{lem:tscheb}Let $\tau_1 \neq \tau_2$ be positive numbers. The system
	\begin{equation}
	\Big(f_1(x), f_2(x), f_3(x), f_4(x)\Big)=\left(\tfrac{1}{\tau_1}e^{-\frac{x}{\tau_1}}, \tfrac{x}{\tau_1^2} e^{-\frac{x}{\tau_1}}, \tfrac{1}{\tau_2}e^{-\frac{x}{\tau_2}}, \tfrac{x}{\tau_2^2}e^{-\frac{x}{\tau_2}}\right)
	\end{equation}
	and its subsystems $(f_1, f_2)$ and $(f_1, f_3, f_4)$ are Tchebycheff systems on $\RR$.
	\end{lemma}
	\begin{proof}
	  By \cite[Ch.~XI, Thm. 1.1]{KS66}, the system $(f_1, f_2, f_3, f_4)$ forms an extended complete Tchebycheff system (ECT-system) on $\RR$, since
        \begin{align*}
            W(f_1)(x) &= \frac{1}{\tau_1}\exp\left(-\frac{x}{\tau_1}\right) > 0\\
            W(f_1, f_2)(x) &= \frac{1}{\tau_1^3}\exp\left(-x\frac{2}{\tau_1}\right) > 0\\
            W(f_1, f_2, f_3)(x) &= \frac{(\tau_2 - \tau_1)^2}{\tau_1^5 \tau_2^3}\exp\left(-x\left(\frac{2}{\tau_1} + \frac{1}{\tau_2}\right)\right) > 0\\
            W(f_1, f_2, f_3, f_4)(x) &= \frac{(\tau_2 - \tau_1)^4}{\tau_1^7 \tau_2^7}\exp\left(-x\left(\frac{2}{\tau_1} + \frac{2}{\tau_2}\right)\right) > 0.
        \end{align*}
        It is therefore in particular a Tchebycheff system. From the calculations we conclude that also $(f_1, f_2)$ is a Tchebycheff system. It furthermore holds that
        \begin{align*}
        	W(f_1, \pm f_3)(x) &= \pm\frac{(\tau_2 - \tau_1)}{\tau_1^2\tau_2^2}\exp\left(-x\left(\frac{1}{\tau_1} + \frac{1}{\tau_2}\right)\right)\\
        	W(f_1, \pm f_3, \pm f_4)(x) &= \frac{(\tau_2 - \tau_1)^2}{\tau_1^3 \tau_2^4}\exp\left(-x\left(\frac{2}{\tau_1} + \frac{1}{\tau_2}\right)\right) > 0,
        \end{align*}
        hence (in dependency on the sign of $\tau_2 - \tau_1$) one of the systems $(f_1, f_3, f_4)$ and $(f_1, -f_3, -f_4)$ is an ECT-system on $\RR$. Again, it follows that $(f_1, f_3, f_4)$ is a Tchebycheff system.
	\end{proof}
	The following Lemma on the relationship between the envelopes and the Wronskians associated to the forward curve and to the yield curve was shown (with exception of claims (d) and (f)) in \cite{KRS23} in the context of the two-dimensional Vasicek model. The proof, however, was only based on the family of lines defined through \eqref{eq:line} and the assumptions discussed in Section~\ref{sub:assumptions}
	    \begin{lemma}[See also Lemma~A.1 in \cite{KRS23}]\label{lem:fw_to_yield}
    Suppose that \ref{item:nonzero}--\ref{item:Wbc} hold true for $\cF^{\rm f}_{(0,T)}$ with $T \in (0,\infty]$. Then the following holds for the envelopes $\eta_{\rm f}$ and $\eta_{\rm y}$ and for the Wronskians associated to the forward curve and to the yield curve: 
        \begin{enumerate}[(a)]
    \item The envelopes $\eta_{\rm f}$ and $\eta_{\rm y}$ start in the same point, i.e., $\eta_{\rm f}(0) = \eta_{\rm y}(0)$;
    \item The endpoint of $\eta_{\rm y}(T)$ is the intersection of $\ell_T^{\rm f}$ and $\ell_T^{\rm y}$;
     \item Every cusp point of $\eta_{\rm y}$ is incident to $\eta_{\rm f}$.
    \item The Wronskian $W(b_{\rm{y}}, c_{\rm{y}})$ has no zeroes on $(0,T)$ and has the same sign as $W(b_{\rm{f}}, c_{\rm{f}})$;
    \item The number of zeros of $W(a_{\rm{y}}, b_{\rm{y}}, c_{\rm{y}})$ is less or equal than the number of zeros of $W(a_{\rm{f}}, b_{\rm{f}}, c_{\rm{f}})$ on $(0,T)$
    \item $W(a_{\rm{y}}, b_{\rm{y}}, c_{\rm{y}})(0) = \tfrac{1}{24} W(a_{\rm{f}}, b_{\rm{f}}, c_{\rm{f}})(0)$
    \end{enumerate}
    \end{lemma}
     \begin{proof}
     All parts, except (d) and (f) have been shown in Lemma~A.1 in \cite{KRS23}. Regarding part (d), we obtain from the calculations in the proof of \cite[Lem. A.1]{KRS23} and \cref{eq:forward_to_yield} the relationship
    	\begin{align}\label{eq:Wronskian_y}
    		W(b_{\rm y}, c_{\rm y})(x) = \frac{1}{x}\begin{vmatrix} b_{\rm y}(x) & c_{\rm y}(x) \\ b_{\rm f}(x) & c_{\rm f}(x)\end{vmatrix} = \frac{1}{x^3}\int_0^x\zeta\begin{vmatrix} b_{\rm f}(\zeta) & c_{\rm f}(\zeta) \\ b_{\rm f}(x) & c_{\rm f}(x)\end{vmatrix}\dif \zeta.
    	\end{align}
    	From \cite[Ch.2 §2, Thm. 2.3]{Kar68} and the assumptions that $c_{\rm f}$ and $W(b_{\rm f}, c_{\rm f})$ have no zero in $(0,T)$ we deduce that the determinant under the integral also has no zero and the same sign as $W(b_{\rm f}, c_{\rm f})$. It follows that the same must hold for $W(b_{\rm y}, c_{\rm y})$. Regarding (f), it follows from the calculations in the proof of \cite[Lem. A.1]{KRS23} that 
	\[g_{\rm y}(0) = \frac{g_{\rm f}(0)}{2}, \qquad g'_{\rm y}(0) = \frac{g'_{\rm f}(0)}{3}, \qquad g''_{\rm y}(0) = \frac{g''_{\rm f}(0)}{4}\]
	holds for any choice of $g \in \set{a,b,c}$. Applying this to the Wronskian $W(a_{\rm{y}}, b_{\rm{y}}, c_{\rm{y}})(0)$, claim (f) follows.
    \end{proof}

           	\begin{lemma}\label{lem:Lambert} Let $a, b,c \in \RR$. If $bc \ge 0$, then the equation $(x+a){\rm e}^{bx} = c$ is solved by 
        		\begin{align*}
        			x = \frac{\mathcal{W}_0(bc\cdot{\rm e}^{ab})}{b} - a, 
        		\end{align*}
 where $\mathcal{W}_0$ is the $0$-th branch of the Lambert $\mathcal{W}$-function. If $bc < 0$ and $|bc| \le e^{ab - 1}$, then the solution takes the same form, but using the branch $\mathcal{W}_{-1}$ of the Lambert $\mathcal{W}$-function. In all other cases, there is no real solution.
 \end{lemma}
        	\begin{proof}
        		The Lambert $\mathcal{W}$-function \cite{CG96} is implicitly defined for all $x \ge 1/e$ by
        		\begin{align*}
        			\mathcal{W}(x){\rm e}^{\mathcal{W}(x)} = x,
        		\end{align*}
        		where $x \ge 0$ yields the branch $\mathcal{W}_0$ and $x \in [-1/e,0)$ the branch $\mathcal{W}_{-1}$. Using the chain of equivalences
        		\begin{align*}
        			&(x+a){\rm e}^{bx} = c\\
        			\Leftrightarrow\quad & (bx+ab){\rm e}^{bx + ab} = bc\cdot {\rm e}^{ab}\\
        			\Leftrightarrow\quad & bx + ab = \mathcal{W}(bc\cdot{\rm e}^{ab}),
        		\end{align*}
the desired result follows.
        	\end{proof}

\bibliographystyle{alpha}
\bibliography{references}

\end{document}